\documentclass[reprint, amsmath,amssymb,showpacs,floatfix,longbibliography, aps, twocolumn, superscriptaddress,nofootinbib]{revtex4-1}
\usepackage{CJKutf8}

\usepackage{amsthm}
\newtheorem{theorem}{Theorem}

\newtheorem{lemma}[theorem]{Lemma}

\usepackage[utf8]{inputenc}
\usepackage{amsmath,amssymb,amsfonts,stmaryrd}
\usepackage{physics}
\usepackage{dsfont}
\usepackage{graphicx}
\usepackage{xcolor}
\usepackage{graphicx}
\usepackage{cancel}
\usepackage{natbib}
\usepackage{color}
\usepackage{tikz}
\usepackage{mathrsfs}
\usepackage{relsize}
\usepackage[linktocpage]{hyperref}
\usepackage[all]{xy}
\hypersetup{colorlinks=true,citecolor=blue,linkcolor=blue, urlcolor=blue, breaklinks=true}

\usepackage{bm} 
\usepackage{subfigure} 
\usepackage{environ}
\usepackage{url}
\usepackage[margin=0.75in]{geometry}

\usepackage{mathtools}

\newcommand{\ZZ}{{\mathbb Z}}

\NewEnviron{eqs}{%
\begin{equation}\begin{split}
    \BODY
\end{split}\end{equation}}

\begin{document}
\newcommand{\Yuan}[1]{ { \color{red} (Yu-An: {#1}) }}
\newcommand{\Yijia}[1]{ { \color{blue} (Yijia: {#1}) }}

\begin{CJK*}{UTF8}{gbsn}

\title{Equivalence between fermion-to-qubit mappings in two spatial dimensions}

\author{Yu-An Chen \CJKfamily{bsmi}(陳昱安)}
\email[E-mail: ]{yuanchen@umd.edu}
\affiliation{Department of Physics, Joint Quantum Institute, and Joint Center for Quantum Information and Computer Science, NIST/University of Maryland, College Park, Maryland 20742, USA}
\affiliation{Condensed Matter Theory Center,University of Maryland, College Park, Maryland 20742, USA}

\author{Yijia Xu (许逸葭) }
\email[E-mail: ]{yijia@umd.edu}
\date{\today}
\affiliation{Department of Physics, Joint Quantum Institute, and Joint Center for Quantum Information and Computer Science, NIST/University of Maryland, College Park, Maryland 20742, USA}
\affiliation{Institute for Physical Science and Technology, University of Maryland, College Park, Maryland 20742, USA}

\begin{abstract}
We argue that all locality-preserving mappings between fermionic observables and Pauli matrices on a two-dimensional lattice can be generated from the exact bosonization in Ref.~\cite{CKR18}, whose gauge constraints project onto the subspace of the toric code with emergent fermions. Starting from the exact bosonization and applying Clifford finite-depth generalized local unitary (gLU) transformation, we can achieve all possible fermion-to-qubit mappings (up to the re-pairing of Majorana fermions).
In particular, we discover a new super-compact encoding using 1.25 qubits per fermion on the square lattice, which is lower than any method in the literature. We prove the existence of fermion-to-qubit mappings with qubit-fermion ratios $r=1+ \frac{1}{2k}$ for positive integers $k$, where the proof utilizes the trivialness of quantum cellular automata (QCA) in two spatial dimensions.
When the ratio approaches 1, the fermion-to-qubit mapping reduces to the 1d Jordan-Wigner transformation along a certain path in the two-dimensional lattice.
Finally, we explicitly demonstrate that the Bravyi-Kitaev superfast simulation, the Verstraete-Cirac auxiliary method, Kitaev’s exactly solved model, the Majorana loop stabilizer codes, and the compact fermion-to-qubit mapping can all be obtained from the exact bosonization.
\end{abstract}

\maketitle
\end{CJK*}

\tableofcontents

\section{Introduction}
A fermion-to-qubit mapping is a duality between local even\footnote{We only consider terms respecting the fermion-parity symmetry, i.e., products with even numbers of fermionic creation and annihilation operators.} fermionic operators and local products of Pauli matrices. 
It is well known that any fermionic system in a 1d lattice can be mapped onto a 1d spin system by the Jordan-Wigner transformation. The Jordan-Wigner transformation can also be applied to systems in higher dimensions by choosing a particular ordering of fermions; however, the mapping becomes highly non-local. From both theoretical and practical points of view, mapping local fermionic operators to local spin operators in higher dimensions is an essential topic.
In the last two decades, there have been many proposals of fermion-to-qubit mappings for two dimensions \cite{verstraete_cirac, kitaev_honeycomb, Whitfield_local_spin_2016, CKR18, majorana_loop, Whitfield_fenwick_2019, Ruba_bosonization_2020, compact_fermion, po_symmetric_2021} and three or arbitrary dimensions \cite{bravyi_kitaev, CK19, C20}. These fermion-to-qubit mappings play important roles in various topics of modern physics, such as exactly solvable models for topological phases \cite{kitaev_honeycomb, EF19, CET21, CH21}, fermionic quantum simulations \cite{bravyi_kitaev, verstraete_cirac, Whitfield_local_spin_2016, majorana_loop, compact_fermion}, and quantum error correction \cite{bravyi2010majorana,vijay2016physical,Fu_Majorana_code1,Fu_Majorana_code2,QC_Majorana_codes,hastings2017small}.
In particular, the exact bosonizations in Refs.~\cite{CKR18, CK19, C20, Ruba_bosonization_2020} construct the toric code with fermions in arbitrary dimensions and impose gauge constraints to restrict in the subspace with emergent fermions, which provide an elegant spacetime description by the Chern-Simons and the Steenrod square topological action. The spacetime pictures for other fermion-to-qubit mappings mentioned above are not manifest. Ref.~\cite{Whitfield_local_spin_2016} points out that the Verstraete-Cirac auxiliary method \cite{verstraete_cirac} can be related to a topological model (toric code), and the compact encoding \cite{compact_fermion} found that its stabilizer is similar to a toric code.

From the theoretical perspective, it is tempting to ask a question: are all fermion-to-qubit mappings in two spatial dimensions ``equivalent'' to the exact bosonization? First, we define the ``equivalence'' by finite-depth generalized local unitary (gLU) transformations. Informally speaking, finite-depth gLU transformation is a finite-depth quantum circuit (FDQC) with ancilla qubits. 
We argue that the answer to the above question is ``yes'' and demonstrate it with examples.

From the practical point of view, fermion-to-qubit mappings are widely used in fermionic quantum simulations of physical systems. In practical quantum simulations, an important quantity is the qubit-fermion ratio $r$, the number of qubits to simulate one fermion on average, since it is directly related to the total number of fermionic modes encoded in a qubit array. Suppose we encode $n$ fermionic modes by $m$ qubits, then the qubit-fermion ratio is $\frac{m}{n}$. The best fermion-to-qubit mapping on the 2d square lattice is the compact fermion-to-qubit mapping with the ratio $r=1.5$ \cite{compact_fermion}.

In this work, we focus on lattices in two spatial dimensions. First, we construct a new super-compact fermion-to-qubit mapping with the qubit-fermion ratio $r=1.25$ on the 2d square lattice. Moreover, we provide a systematic approach to construct various 2d bosonizations by utilizing the ideas of Clifford circuit \cite{Gottesman_Clifford,gottesman1998heisenberg} and finite-depth generalized local unitary (gLU) transformations \cite{Xiechen_local_unitary,QI_meets_QM}. Such an approach provides a new perspective to study the relationship between different fermion-to-qubit mappings. We find that all the local fermion-to-qubit mappings can be generated from the exact bosonization by finite-depth gLU transformations. In particular, we explicitly show how to obtain the Bravyi-Kitaev superfast encoding (BKSF), the Verstraete-Cirac mapping, Kitaev's honeycomb model, the Majorana loop stabilizer codes (MLSC), and the compact fermion-to-qubit mapping.

\subsection*{Summary of results}
We first demonstrate a super-compact fermion-to-qubit mapping on the 2d square lattice with qubit-fermion ratio $r=1.25$ in Sec.~\ref{sec:super_compact} and compare its data with other fermion-to-qubit mappings in Table~\ref{tab: results compared with previous papers}.
In Sec.~\ref{sec:gLU_and_bosonization}, we define the crucial theoretical technique in our construction: the generalized local unitary (gLU) transformation \cite{Xiechen_local_unitary,QI_meets_QM}.
Then, in Sec.~\ref{sec:square_r_1.5}, we derive the $r=1.5$ fermion-to-qubit mapping, which is equivalent to the compact encoding \cite{compact_fermion}. In Sec.~\ref{sec:square_r_1.25}, we further improve the ratio to derive the $r=1.25$ construction shown in the previous section. 
In Sec.~\ref{sec:Clifford_construction}, we prove that a general construction with ratio $r=1 + \frac{1}{2k}$ exists for any positive integer $k$. The proof utilizes the trivialness of 2d quantum cellular automata (QCA) \cite{FH20,Haah_Clifford_QCA21}.
In Sec.~\ref{sec:equivalence_relation}, we define the equivalence relation between different 2d bosonizations based on finite-depth gLU transformations and discuss the equivalence between the exact bosonization and many well-known fermion-to-qubit mappings. We construct explicit Clifford circuits that convert the exact bosonization to the Bravyi-Kitaev superfast encoding (Sec.~\ref{sec:Bravyi_Kitaev}), the Verstraete-Cirac mapping (Sec.~\ref{sec:Verstraete_Cirac}), Kitaev's honeycomb model (Sec.~\ref{sec:Kitaev_honeycomb}), the Majorana loop stabilizer codes (Sec.~\ref{sec:Majorana_loop}), and the Jordan-Wigner transformation (Sec.~\ref{sec:Jordan_Wigner}).

\begin{table*}[htbp]
    \centering
    \begin{tabular}{ |c | c | c | c | c |}
    \hline
    {}& qubit-fermion-ratio $r$ &fermion parity weight& hopping weight & stabilizer weight  \\
    \hline
    Verstraete-Cirac \cite{verstraete_cirac} \footnote{The graph structure of the auxiliary Hamiltonian is Fig.~\ref{fig:VC_auxiliary_ordering}.}& 2& 1& 3-4 & 6\\
    \hline
    BKSF \cite{bravyi_kitaev} \footnote{The ordering of edges is shown in Fig.~\ref{fig:BKSF_labeling}.} &2 & 4 & 2-6 & 6 \\
    \hline
    Kitaev's honeycomb model \cite{kitaev_honeycomb} &2&2 & 2-5 &6 \\
    \hline
    Exact bosonization \cite{CKR18} &2 & 4 & 2-6 & 6\\
    \hline
    MLSC \cite{majorana_loop} &2&3 & 3-4 &4-10 \\
    \hline
    Compact fermion-to-qubit mapping \cite{compact_fermion} &1.5 & 1 &3 & 8\\
    \hline
    Super-compact fermion-to-qubit mapping & 1.25 & 1-2 & 2-6 & 12\\
    \hline
    \end{tabular}
    \caption{\label{tab: results compared with previous papers}
    Comparison between fermion-to-qubit mappings on the 2d square lattice.
    }
\end{table*}

\section{Super-compact fermion-to-qubit mapping}\label{sec:super_compact}

In this section, we introduce a super-compact encoding of fermions by qubits with the qubit-fermion ratio $r = 1.25$. We first introduce the Hilbert spaces for fermions and qubits and then describe the mapping between them.

On the 2d square lattice in Fig.~\ref{fig:square_deformed_lattice}, each vertex $v$ contains a fermionic mode with creation/annihilation operator $c_v^\dagger$, $c_v$ with the standard commutation relation $\{ c_v, c^\dagger_{v^\prime} \} = \delta_{v v'}$. It is easier to use the Majorana basis
\begin{eqs}
    \gamma_v=c_v+c_v^\dagger,\quad \gamma_v'=\frac{c_v-c_v^\dagger}{i}.
\end{eqs}
The local fermion parity operator at a vertex $v$ is
\begin{eqs}
    B_v \equiv (-1)^{c_v^\dagger c_v} = - i \gamma_v \gamma^\prime_v,
\end{eqs}
and the hopping operator on an edge $e$ is
\begin{eqs}
    A_e = i\gamma_{L(e)}\gamma_{R(e)},
\end{eqs}
where $L(e)$ and $R(e)$ are the left and right vertices of the edge $e$ defined in Fig.~\ref{fig:hopping_direction}.
The even algebra of fermions consists of local observables with a trivial fermion parity, i.e., local observables which commute with the total fermion parity $(-1)^F \equiv \prod_f (-1)^{c^\dagger_f c_f}$.\footnote{The even fermionic algebra can also be considered as the algebra of local observables containing an even number of Majorana operators.}
The generators for the even algebra of fermions are $A_e$ and $B_v$ on all edges and vertices \cite{CKR18}.

\begin{figure}[t]
    \centering
    \includegraphics[width=0.35\textwidth]{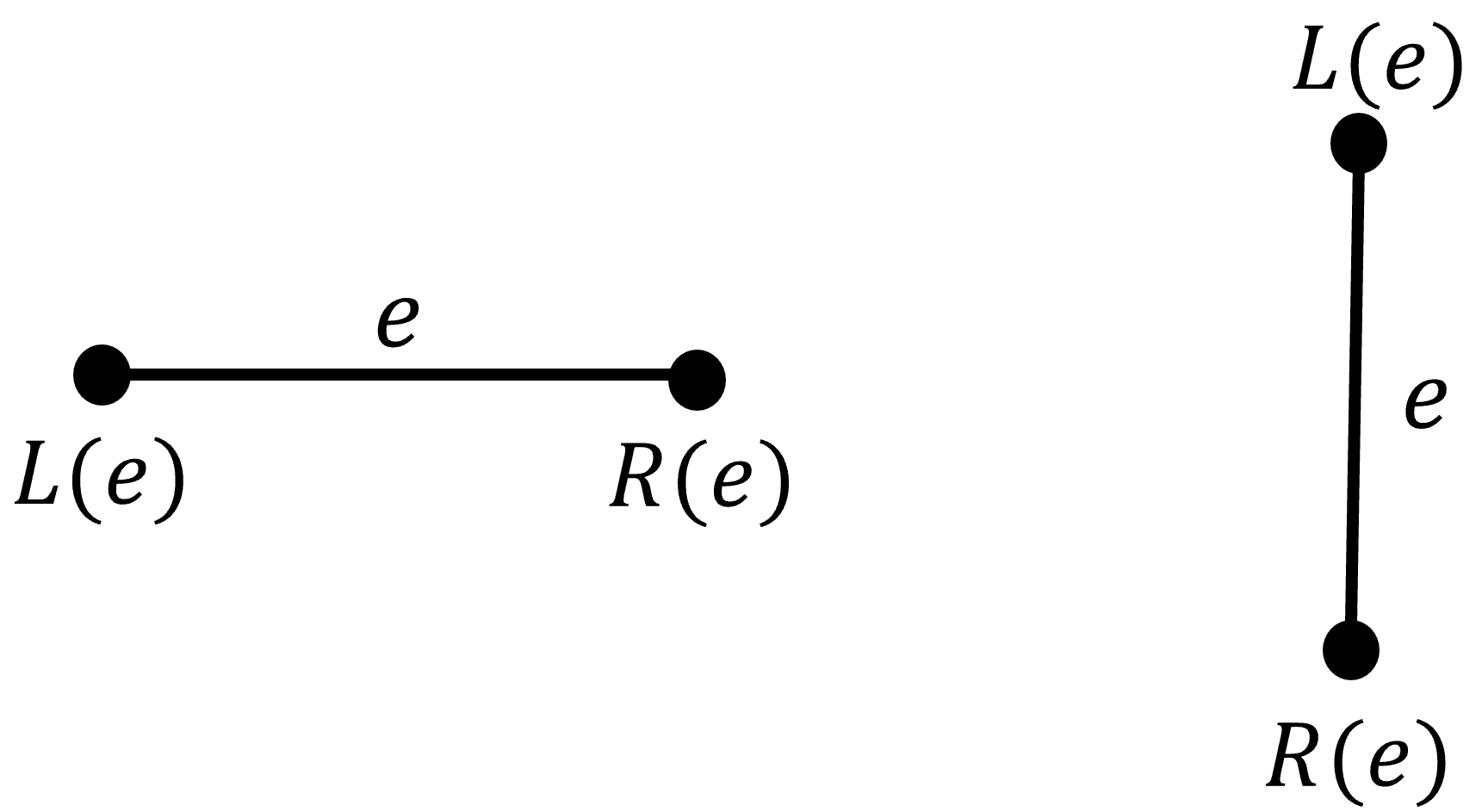}
    \caption{Definition of left and right along horizontal and vertical edges}
    \label{fig:hopping_direction}
\end{figure}

\begin{figure}[h]
    \centering
    \includegraphics[width=0.45\textwidth]{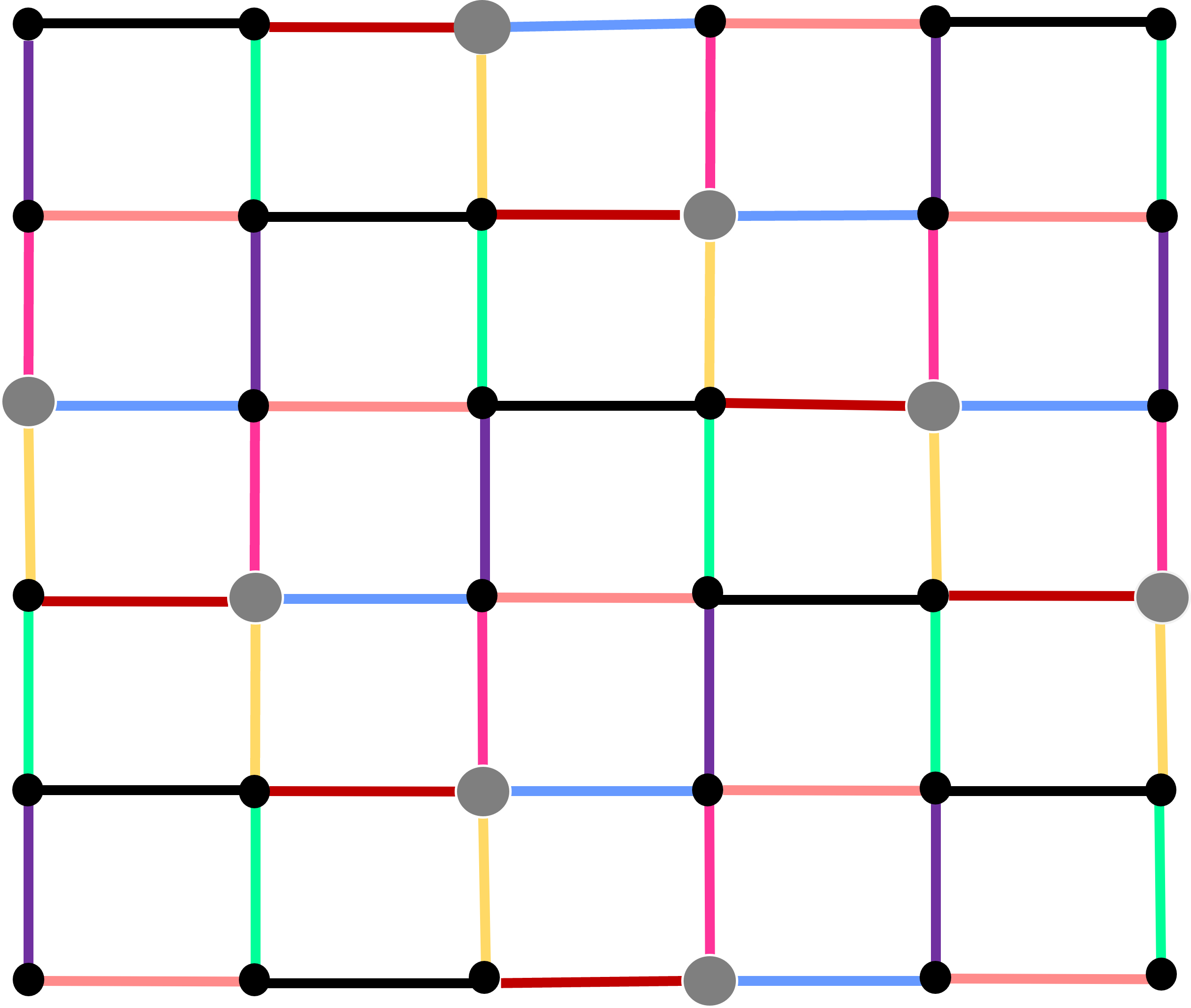}
    \caption{The Hilbert space for the super-compact encoding. Each vertex encodes a fermionic mode, while each black vertex has 1 qubit and each grey vertex has 2 qubits. The qubit-fermion ratio $r$ is $1.25$ in this setting.}
    \label{fig:square_deformed_lattice}
\end{figure}

On the other hand, the qubits are put at vertices in Fig.~\ref{fig:square_deformed_lattice}. We label vertices by black and grey colors. For each black vertex, there is 1 qubit, and for each grey vertex, there are 2 qubits. As shown in Fig.~\ref{fig:square_logical} and \ref{fig:square_stabilizer}, each grey vertex has two Pauli matrices on the top-right and bottom-left corners respectively.

The fermion-to-qubit mappings are mappings from $A_e, B_v$ to Pauli strings (products of Pauli matrices) on qubits with the same algebra. In addition, such mapping satisfies a condition that the product of $A_e$ along an arbitrary closed path should be the identity operator (up to a phase) since all Majorana operators cancel out. Such constraint requires the qubit system to be stabilized by a stabilizer group which is a set of hopping operators along with all closed loops.
Now, we explicitly construct the mapping on the lattice in Fig.~\ref{fig:square_deformed_lattice}:
\begin{eqs}
    A_e &= i\gamma_{L(e)}\gamma_{R(e)} \longleftrightarrow \tilde A_e, \\
    B_v &= - i \gamma_v \gamma^\prime_v \longleftrightarrow \tilde B_v,
\end{eqs}
where $\tilde A_e$ and $\tilde B_v$ are defined in Fig.~\ref{fig:square_logical}. It can be checked that two operators $\tilde A_e$ and $\tilde A_{e'}$ anti-commute if and only if $e$ and $e'$ are two distinct edges sharing one common vertex, and $\tilde A_e$ and $\tilde B_v$ anti-commute if and only if the edge $e$ contains the vertex $v$. Therefore, $\{ A_e, B_v \}$ and $\{ \tilde A_e, \tilde B_v \}$ satisfy the same commutation relations. The last step is to impose the stabilizer conditions (gauge constraints) that product of $\tilde A_e$ on a loop $l$ is proportional to the identity operator:
\begin{eqs}
    \prod_{e \in l} \tilde A_e = i^{|l|},
\label{eq: product of tilde Ae}
\end{eqs}
where $|l|$ is the length of the loop $l$.

The generators of this stabilizer condition are expressed in
Fig.~\ref{fig:square_stabilizer}. The stabilizers for the vertices that connect pink, black, green, purple edges generate the whole stabilizer group. The weight of such a stabilizer is $12$.

\begin{figure*}
    \centering
    \includegraphics[width=0.75\textwidth]{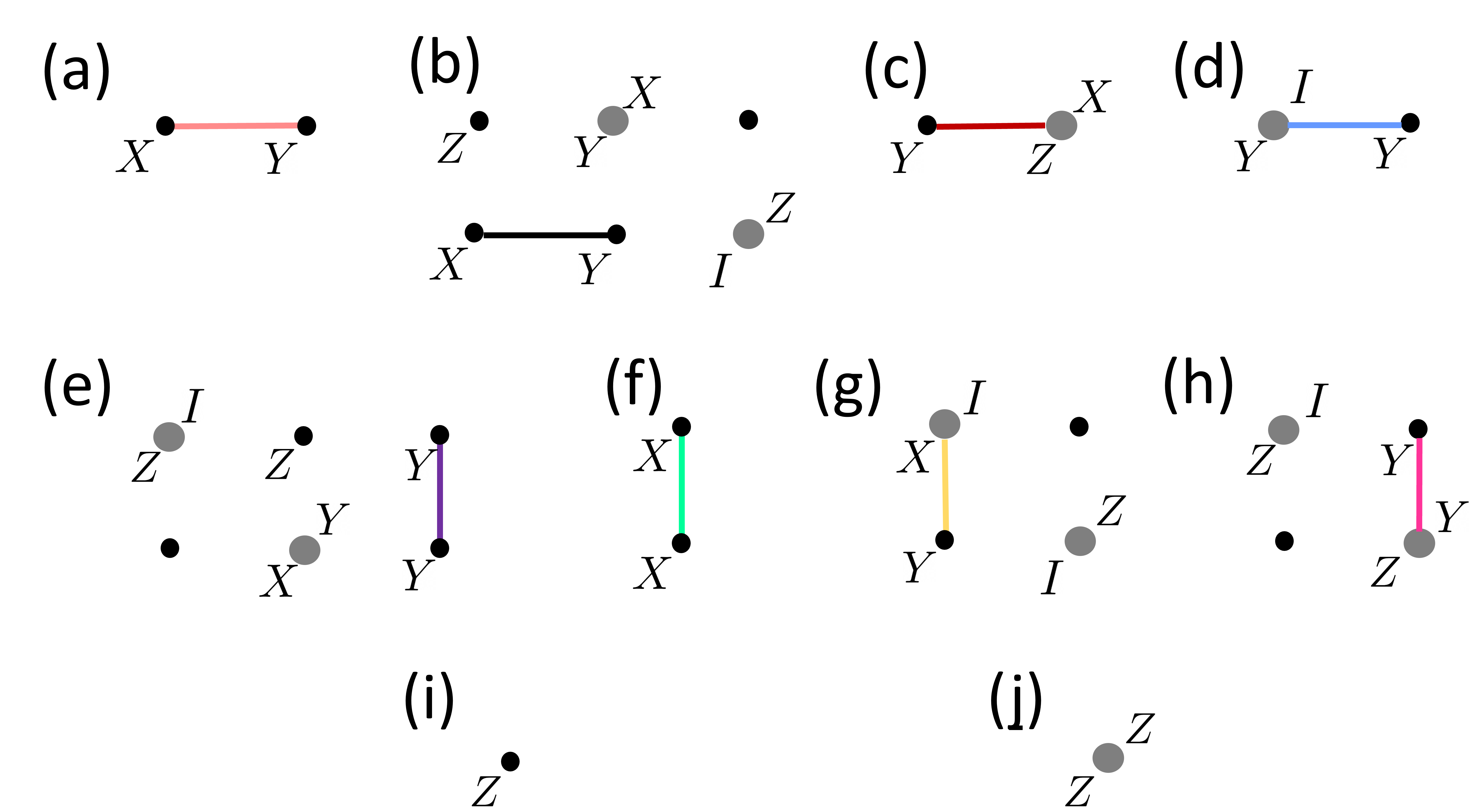}
    \caption{The hopping term $\tilde A_e$ and the parity term $\tilde B_v$ in the bosonic Hilbert space. The definition of $\tilde A_e$ and $\tilde B_v$ depend on the colors of edges and vertices. (a),(b),(c),(d) are four kinds of horizontal hopping terms; (e),(f),(g),(h) are four kinds of vertical hopping terms; (i), (j) are parity terms on black and grey vertices.}
    \label{fig:square_logical}
\end{figure*}

\begin{figure}
    \centering
    \includegraphics[width=0.4\textwidth]{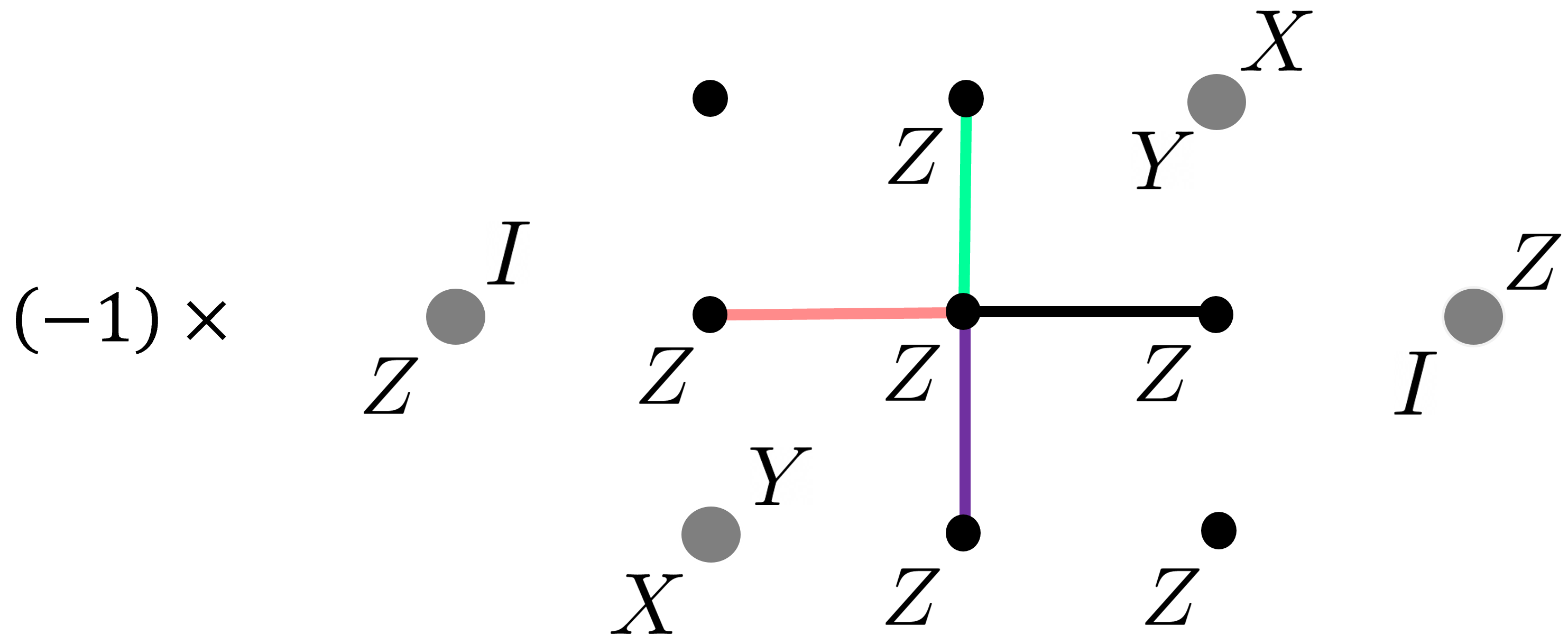}
    \caption{The stabilizer acts on the vertex that connects to pink, black, purple, and green edges. The product of $\tilde A_e$ on any closed loop is generated by this stabilizer.}
    \label{fig:square_stabilizer}
\end{figure}

\section{Generalized local unitary circuits on the exact bosonization}\label{sec:gLU_and_bosonization}

\begin{figure}[h]
    \centering
    \includegraphics[width=0.3\textwidth]{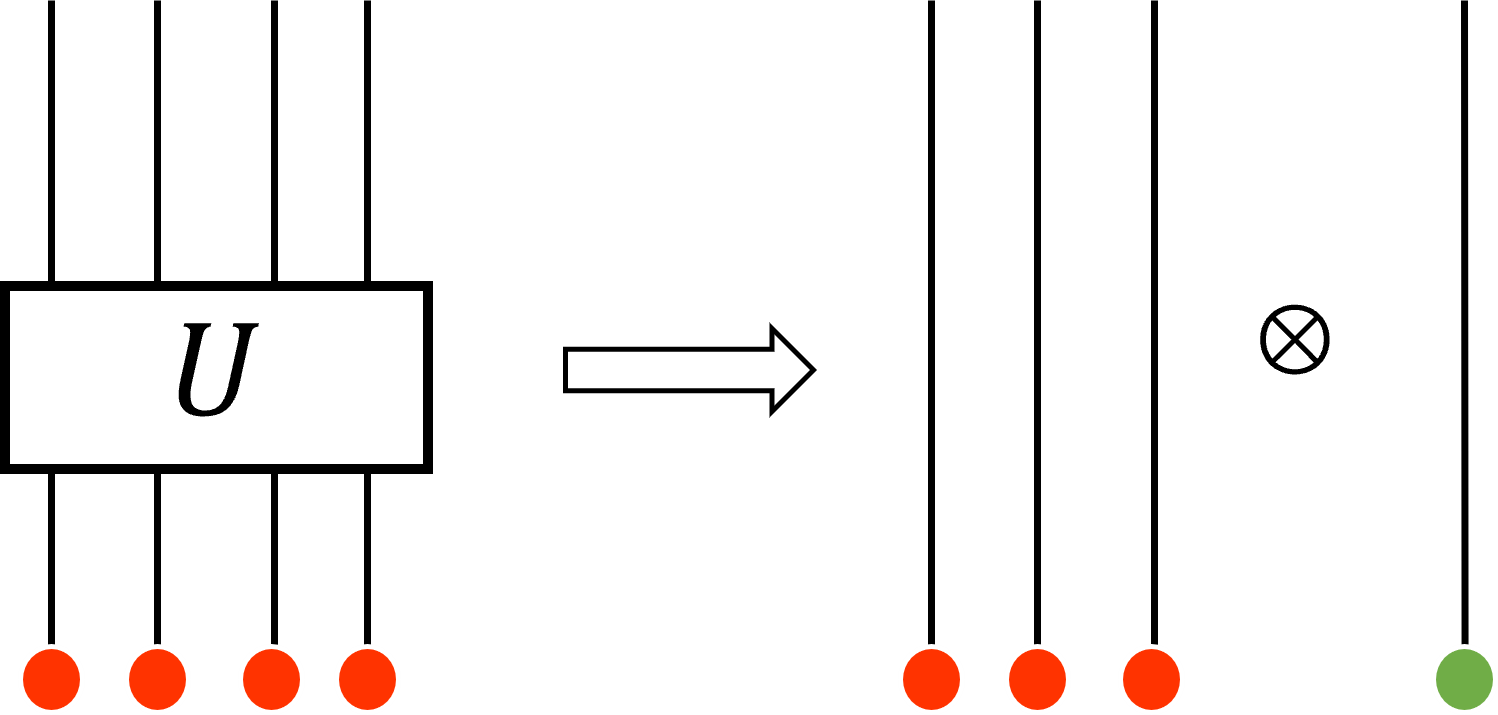}
    \caption{We disentangle the green qubit from others by a local unitary transformation $U$ and then discard this part. This is called a generalized local unitary circuit.}
    \label{}
\end{figure}

We describe a systematical way to derive fermion-to-qubit mappings from the exact bosonization in two spatial dimensions in the section. In Ref.~\cite{CKR18}, the exact bosonization is proposed using the subspace of the toric code with fermionic excitations, which will be reviewed in Sec.~\ref{sec: review of 2d bosonization}. By applying local unitary operators on the exact bosonization, we can generate a new fermion-to-qubit mapping. However, to include the lattice deformation or changing of the Hilbert spaces, local unitary operators are not sufficient, and the idea of generalized local unitary (gLU) operators is introduced.

The generalized local unitary (gLU) \cite{Xiechen_local_unitary,QI_meets_QM} arises from the idea of wave function renormalization where local unitary operators are used to add or remove degrees of freedom at different length scales. For a wave function, we can use an operation to add or remove ancilla qubits in the product state.\footnote{Since the ancilla qubits are in the products, adding or removing them do not change the information contained in the state.} We call the transformation that changes the degrees of freedom a generalized local unitary (gLU) operator.
The formal definition of gLU in Ref.~\cite{Xiechen_local_unitary} is as follows. For a quantum state $\ket{\Phi}$ with the reduced density matrix $\rho_A$ in region $A$, $\rho_A$ only acts in a support subspace $V_A^{sp}$ of the total Hilbert space $\mathscr{H}_A$. The dimension of $V_A^{sp}$ is $D_A^{sp}$, which is called support dimension. Hence, the total Hilbert space on region $A$ can be written as a direct sum $\mathscr{H}_A=V_A^{sp} \oplus \overline{V}_A^{sp}$. Let $\ket{\Tilde{\psi}_i}$, $i=1, ..., D_A^{sp}$ to be the basis of $V_A^{sp}$, $\ket{\Tilde{\psi}_i}$, $i=D_A^{sp}+1,...,D_A$ to be the basis of $\overline{V}_A^{sp}$, and
$\ket{\psi_i}$, $i=1,...,D_A$ to be the basis of $\mathscr{H}_A$ ($D_A= \mathrm{dim} (\mathscr{H}_A)$). We introduce the local unitary transformation $U^{full}$ to rotate $\ket{\psi_i}$ to $\ket{\Tilde{\psi}_i}$. In the new basis, wave function $\ket{\Phi}$ only has non-zero amplitudes on the first $D_A^{sp}$ basis vectors, and therefore we can truncate out the remaining columns of $U^{full}$ to get the gLU operator $U$ without losing any information.

With gLU transformation, we can remove the degrees of freedom in the system if they are in the product states. This operation is equivalent to disentangling parts of qubits from others. Hence, the qubit-fermion ratio $r$ can be improved by wisely applying finite-depth gLU to the exact bosonization. In this paper, we will use finite-depth gLU Clifford circuits since we focus on Pauli stabilizer models. We demonstrate the construction of fermion-to-qubit mappings with ratio $r=1.5$ in Sec.~\ref{sec:square_r_1.5} and $r=1.25$ in Sec.~\ref{sec:square_r_1.25} by conjugating the 2d exact bosonization by certain finite-depth gLU Clifford circuits.

\subsection{Review of the exact bosonization}\label{sec: review of 2d bosonization}

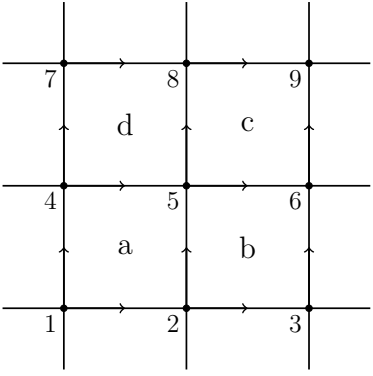
\begin{figure}[htb]
\centering
\resizebox{5cm}{!}{%
\begin{tikzpicture}
\draw[thick] (-3,0) -- (3,0);\draw[thick] (-3,-2) -- (3,-2);\draw[thick] (-3,2) -- (3,2);
\draw[thick] (0,-3) -- (0,3);\draw[thick] (-2,-3) -- (-2,3);\draw[thick] (2,-3) -- (2,3);
\draw[->] [thick](0,0) -- (1,0);\draw[->][thick] (0,2) -- (1,2);\draw[->][thick] (0,-2) -- (1,-2);
\draw[->][thick] (0,0) -- (0,1);\draw[->][thick] (2,0) -- (2,1);\draw[->][thick](-2,0) -- (-2,1);
\draw[->][thick] (-2,0) -- (-1,0);\draw[->][thick] (-2,2) -- (-1,2);\draw[->][thick](-2,-2) -- (-1,-2);
\draw[->][thick] (-2,-2) -- (-2,-1);\draw[->][thick] (0,-2) -- (0,-1);\draw[->] [thick](2,-2) -- (2,-1);
\filldraw [black] (-2,-2) circle (1.5pt) node[anchor=north east] {\large 1};
\filldraw [black] (0,-2) circle (1.5pt) node[anchor=north east] {\large 2};
\filldraw [black] (2,-2) circle (1.5pt) node[anchor=north east] {\large 3};
\filldraw [black] (-2,-0) circle (1.5pt) node[anchor=north east] {\large 4};
\filldraw [black] (0,0) circle (1.5pt) node[anchor=north east] {\large 5};
\filldraw [black] (2,0) circle (1.5pt) node[anchor=north east] {\large 6};
\filldraw [black] (-2,2) circle (1.5pt) node[anchor=north east] {\large 7};
\filldraw [black] (0,2) circle (1.5pt) node[anchor=north east] {\large 8};
\filldraw [black] (2,2) circle (1.5pt) node[anchor=north east] {\large 9};
\draw (-1,-1) node{\Large a};
\draw (1,-1) node{\Large b};
\draw (1,1) node{\Large c};
\draw (-1,1) node{\Large d};
\end{tikzpicture}
}
\caption{Bosonization on a square lattice \cite{CKR18}. We put Pauli matrices $X_e$, $Y_e$, $Z_e$ on each edge and one complex fermion $c_f, c_f^{\dagger}$ at each face. We work on the Majorana basis $\gamma_f = c_f + c_f^\dagger$ and $\gamma'_f = -i (c_f - c_f^\dagger)$ for convenience.}
\label{fig:square}
\end{figure}

We review the exact bosonization on the Hilbert space defined in Fig.~\ref{fig:square}.
The elements of vertices, edges, and faces are denoted $v,e,f$. On each face $f$ of the lattice we place a single pair of fermionic creation-annihilation operators $c_f,c_f^\dagger$, or equivalently a pair of Majorana fermions $\gamma_f,\gamma'_f$. The even fermionic algebra consists of local observables with a trivial fermion parity, i.e., local observables which commute with the total fermion parity $(-1)^F \equiv \prod_f (-1)^{c^\dagger_f c_f}$.
The even algebra is generated by \cite{CKR18}:
\begin{enumerate}
    \item On-site fermion parity:
    \begin{equation}\label{eq: fermion_parity}
    P_f \equiv -i\gamma_f\gamma'_f.
\end{equation}
    \item Fermionic hopping term:
    \begin{equation}\label{eq: fermion_hopping}
    S_e \equiv i\gamma_{L(e)}\gamma'_{R(e)},
    \end{equation}
    where $L(e)$ and $R(e)$ are faces to the left and right of $e$, with respect to the orientation of $e$ in Fig.~\ref{fig:square}.
\end{enumerate}

The bosonic dual of this system involves $\ZZ_2$-valued spins on the edges of the square lattice. For every edge $e$, we define a unitary operator $U_e$ that squares to $1$. Labeling the faces and vertices as in Fig. \ref{fig:square}, we define:
\begin{equation}
\begin{split}
U_{56} &= X_{56} Z_{25}, \\
U_{58} &= X_{58} Z_{45}, 
\end{split}
\label{eq: Ue definition}
\end{equation}
where $X_e$, $Z_e$ are Pauli matrices acting on a spin at each edge $e$.
Operators $U_e$ for other edges are defined by using translation symmetry. Pictorially, the operator $U_e$ is drawn as
\begin{equation}
    U_e=\begin{gathered}
    \xymatrix@=1cm{&{}\ar@{-}[d]|{\displaystyle X_e}\\
    {}\ar@{-}[r]|{\displaystyle Z}&}
    \end{gathered} \qquad \text{or} \qquad \begin{gathered}
    \xymatrix@=1cm{{}\ar@{-}[r]|{\displaystyle X_e}&\\
    {}\ar@{-}[u]|{\displaystyle Z}&}
    \end{gathered},
\label{eq: original Ue}
\end{equation}
corresponding to the vertical or horizontal edge $e$.

It has been shown in Ref.~\cite{CKR18} that $U_e$ and $S_e$ satisfy the same commutation relations. We also map the fermion parity $P_f$ at each face $f$ to the ``flux operator'' $W_f \equiv \prod_{e \subset f} Z_e$, the product of $Z_e$ around a face $f$:
\begin{equation}
    W_f=\begin{gathered}
   \xymatrix@=1cm{%
    {}\ar@{-}[r]|{\displaystyle Z}\ar@{}[dr]|{\mathlarger f} & {}\ar@{-}[d]|{\displaystyle Z} \\
    {}\ar@{-}[u]|{\displaystyle Z} & {}\ar@{-}[l]|{\displaystyle Z}}
    \end{gathered}.
\label{eq: original Wf}
\end{equation}
The bosonization map is
\begin{equation}
\begin{split}
    S_e  &\longleftrightarrow U_e,\\
    P_f  &\longleftrightarrow W_f,
\end{split}
\label{eq: 2d bosonization map}
\end{equation}
or pictorially 
\begin{align}
    i\times
    \begin{gathered}
    \xymatrix@=1.2cm{
    &\\
    {}\ar@{}[ur]|{\displaystyle \gamma_{L(e)}} \ar@{}[dr]|{\displaystyle \gamma'_{R(e)}} \ar@{-}[r]^{ \mathlarger e} &  \\
    &
    }
    \end{gathered}
    &\begin{gathered}\xymatrix{
    {}\ar@{<->}[r] &{}}\end{gathered}
    \hspace{0.3cm}\begin{gathered}
    \xymatrix@=1cm{%
    {}\ar@{-}[r]|{\displaystyle X_e} &{} \\{}\ar@{-}[u]|{\displaystyle Z}& {}}
    \end{gathered} \quad ,
    \label{eq: Ue on horizontal edge}
    \\[-15pt]
    i\times
    \begin{gathered}
    \xymatrix@=1.2cm{
    &{}\ar@{-}[d]^{\mathlarger e}&\\
    \ar@{}[ur]|{\displaystyle \gamma_{L(e)}}&\ar@{}[ur]|{\quad \displaystyle \gamma'_{R(e)}}&}
    \end{gathered}
    &\begin{gathered}\xymatrix{
    {}\ar@{<->}[r] &{}}\end{gathered}
    \begin{gathered}
    \xymatrix@=1cm{%
    &{}\ar@{-}[d]|{\displaystyle X_e} \\{}\ar@{-}[r]|{\displaystyle Z}&
    }
    \end{gathered} \quad ,
    \\[10pt]
    -i\gamma_f\gamma_f'
    &\begin{gathered}\xymatrix{
    {}\ar@{<->}[r] &{}}\end{gathered}
    \hspace{0.3cm}
    \begin{gathered}
    \xymatrix@=1cm{%
    {}\ar@{-}[r]|{\displaystyle Z}\ar@{}[dr]|{\mathlarger f} & {}\ar@{-}[d]|{\displaystyle Z} \\
    {}\ar@{-}[u]|{\displaystyle Z} & {}\ar@{-}[l]|{\displaystyle Z}}
    \end{gathered} \quad .
\end{align}
The condition $ P_a P_c S_{{58}} S_{{56}} S_{{25}} S_{{45}} = 1$ on fermionic operators gives gauge constraints (stabilizer) $G_v=W_{f_c} \prod_{e \supset v_5} X_e  =1$ for bosonic operators, or generally
\begin{equation}
    G_v =
    \begin{gathered}
   \xymatrix@=1cm{%
    &{}\ar@{-}[r]|{\displaystyle Z}\ar@{}[dr]|{\mathlarger f}& {}\ar@{-}[d]|{\displaystyle Z}\\
    {}\ar@{-}[r]|{\displaystyle X}&{v}\ar@{-}[u]|{\displaystyle XZ} & {}\ar@{-}[l]|{\displaystyle XZ}\\
    &{}\ar@{-}[u]|{\displaystyle X}&}
\end{gathered}
= 1.
\label{eq:gauge constraint at vertex}
\end{equation}
The gauge constraint Eq.~(\ref{eq:gauge constraint at vertex}) can be considered as the stabilizer ($G_v \ket{\Psi} = \ket{\Psi}$ for $\ket{\Psi}$ in the code space), which forms the stabilizer group $\mathcal{G}$. The operators $U_e$ and $W_f$ generate all logical operators.\footnote{The logical operators consist of all operators that commute with $\mathcal{G}$. $\mathcal{G}$ are trivial logical operators since stabilizers have no effect on the code space.
$U_e$ and $W_f$ generate all logical operators.} In the setting above, qubits live on edges and fermions live on faces, so the ratio between the number of qubits and the number of fermions is $r=2$. We are going to apply finite-depth gLU transformations to lower this ratio.

\subsection{Compact fermion-to-qubit mapping with ratio \texorpdfstring{$r=1.5$}{}} \label{sec:square_r_1.5}

In the exact bosonization on the square lattice, the bosonic subspace is constrained by the stabilizer Eq.~\eqref{eq:gauge constraint at vertex} at each vertex. First, we enlarge the unit cell to be a $2\times 2$ square
\begin{eqs}
    \begin{gathered}
        \xymatrix@=1cm{%
        &\ar@{-}[d]|{} \ar@{}[dr]|{\displaystyle \text{even}} \ar@{-}[r]|{} & \ar@{-}[r]|{} \ar@{-}[d]|{} \ar@{}[dr]|{\displaystyle \text{odd}}&\ar@{-}[d]|{} \\
        &\ar@{-}[r]|{}  \ar@{}[dr]|{\displaystyle \text{odd}} \ar@{-}[d]|{} &\ar@{-}[r]|{}  \ar@{}[dr]|{\displaystyle \text{even}} \ar@{-}[d]|{} &  \ar@{-}[d]|{} \\
        &\ar@{-}[r]|{}& \ar@{-}[r]|{} & 
        }
    \end{gathered} \quad.
\label{eq: 2x2 unit square}
\end{eqs}
Note that we have colored the faces to be even or odd as the checkerboard.
In each $2\times 2$ square, there are totally $4$ fermions, $8$ qubits and $4$ stabilizers, whose qubit-fermion ratio is $r = \frac{8}{4}=2$. We are going to apply a finite-depth gLU circuit to disentangle some qubits and reduce the ratio.

\begin{figure*}[t]
    \centering
    \includegraphics[width=0.8\textwidth]{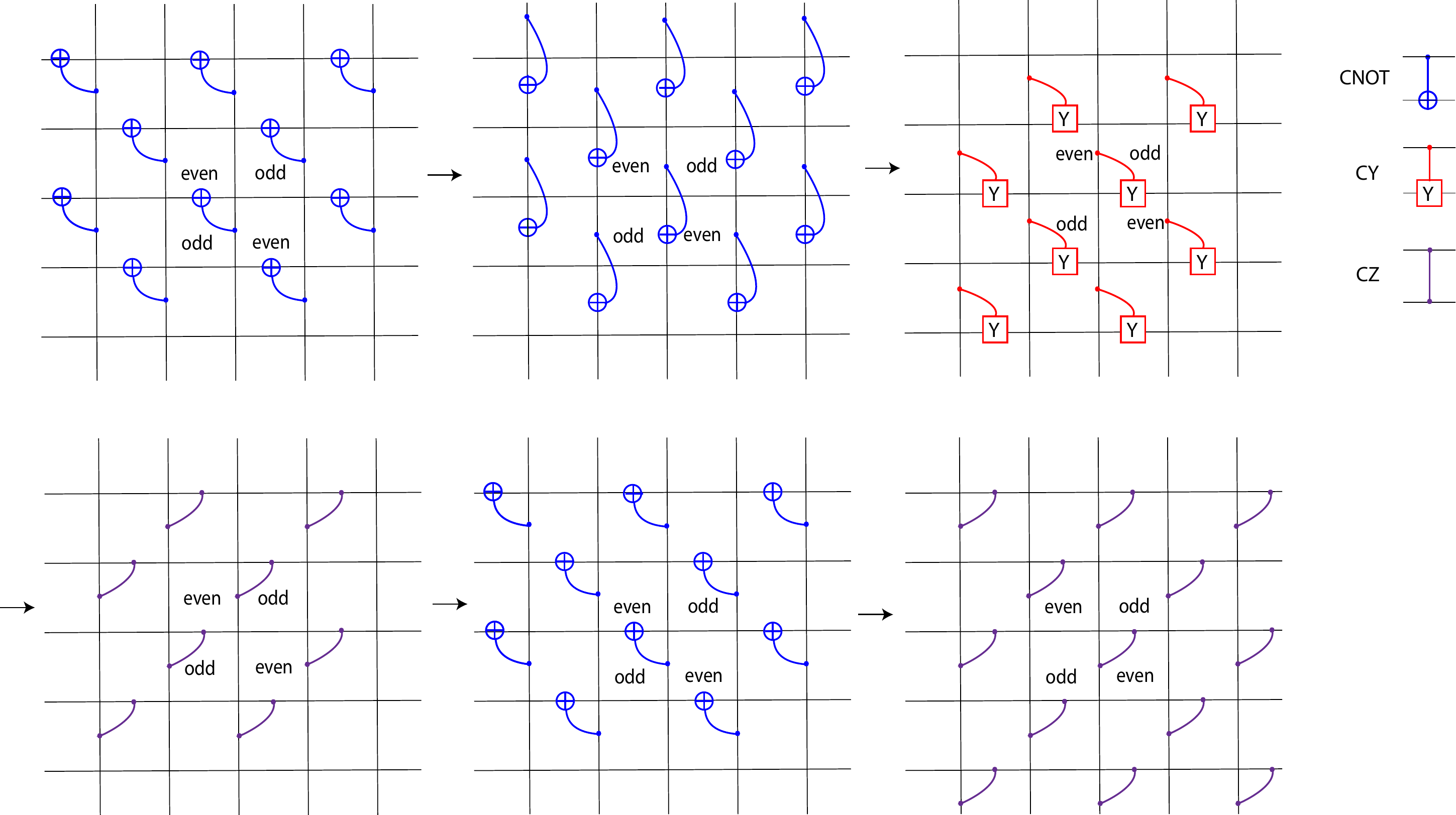}
    \caption{The finite-depth Clifford circuit for the $r=1.5$ construction}
    \label{fig:square_lattice}
\end{figure*}

In Fig.~\ref{fig:square_lattice}, the translational invariant Clifford circuit is defined.\footnote{We have enlarged the unit cell and therefore the distances for the translational generators are doubled. } We divide the stabilizers into two cases, living on an odd face or an even face, as shown below
\begin{eqs}
    G_{\text{odd}}=\begin{gathered}
   \xymatrix@=1cm{%
    &{}\ar@{-}[r]|{\displaystyle Z} & {}\ar@{-}[d]|{\displaystyle Z}\\
    {}\ar@{-}[r]|{\displaystyle X}&{}\ar@{-}[u]|{\displaystyle XZ} \ar@{}[ur]|{\displaystyle ~\text{odd}}  & {}\ar@{-}[l]|{\displaystyle XZ}\\
    &{}\ar@{-}[u]|{\displaystyle X}&}
    \end{gathered},
    \quad
    G_{\text{even}}=
    \begin{gathered}
        \xymatrix@=1cm{%
        &{}\ar@{-}[r]|{\displaystyle Z} & {}\ar@{-}[d]|{\displaystyle Z}\\
        {}\ar@{-}[r]|{\displaystyle X}&{}\ar@{-}[u]|{\displaystyle XZ} \ar@{}[ur]|{\displaystyle ~\text{even}}  & {}\ar@{-}[l]|{\displaystyle XZ}\\
        &{}\ar@{-}[u]|{\displaystyle X}&}
    \end{gathered}. \nonumber
\end{eqs}
After the conjugation of the Clifford circuit in Fig.~\ref{fig:square_lattice}, these stabilizers become
\begin{align}
    \begin{gathered}
    \xymatrix@=1cm{%
        &{}\ar@{-}[r]|{\displaystyle Z} & {}\ar@{-}[d]|{\displaystyle Z}\\
        {}\ar@{-}[r]|{\displaystyle X}&{}\ar@{-}[u]|{\displaystyle XZ} \ar@{}[ur]|{\displaystyle ~\text{odd}}  & {}\ar@{-}[l]|{\displaystyle XZ}\\
        &{}\ar@{-}[u]|{\displaystyle X}&}
    \end{gathered} \quad 
    &\begin{gathered}
        \xymatrix{
            {}\ar[r] &{}}
        \end{gathered} \quad
    (-1)\times \begin{gathered}
    \xymatrix@=1cm{%
        &{}\ar@{-}[r]|{} & {}\ar@{-}[d]|{}\\
        {}\ar@{-}[r]|{}&{}\ar@{-}[u]|{\displaystyle Y} \ar@{}[ur]|{\displaystyle \text{odd}}  & {}\ar@{-}[l]|{}\\
        &{}\ar@{-}[u]|{}&}
    \end{gathered}, \label{eq: r=1_5 stabilizer cancel}\\
    \begin{gathered}
        \xymatrix@=1cm{%
        &{}\ar@{-}[r]|{\displaystyle Z} & {}\ar@{-}[d]|{\displaystyle Z}\\
        {}\ar@{-}[r]|{\displaystyle X}&{}\ar@{-}[u]|{\displaystyle XZ} \ar@{}[ur]|{\displaystyle ~\text{even}}  & {}\ar@{-}[l]|{\displaystyle XZ}\\
        &{}\ar@{-}[u]|{\displaystyle X}&}
    \end{gathered}
    \quad
    &\begin{gathered}
        \xymatrix{
            {}\ar[r] &{}}
    \end{gathered}
    \quad
    (-1)\times \begin{gathered}
        \xymatrix@=1cm{%
        &{}\ar@{-}[r]|{} &{}\ar@{-}[d]|{} \\
        {}\ar@{-}[r]|{\displaystyle X} {}\ar@{-}[d]|{\displaystyle Z}&{}\ar@{-}[u]|{\displaystyle Y} \ar@{}[ur]|{\displaystyle \text{even}}  \ar@{-}[r]|{\displaystyle X}  & \\
        {}\ar@{-}[r]|{\displaystyle X}&{}\ar@{-}[u]|{} \ar@{-}[r]|{\displaystyle X} \ar@{-}[d]|{\displaystyle Y}&{}\ar@{-}[u]|{\displaystyle Z}\\
        &&}
    \end{gathered}.
\label{eq: r=1_5 stabilizer}
\end{align}
We have converted the stabilizer $G_{\text{odd}}$ into a single-qubit stabilizer $Y$.
This qubit will be in the eigenstate of $Y$ and can be discarded. Hence, we successfully eliminate the qubits on the left edges of all odd faces. In the $2 \times 2$ unit square Eq.~\eqref{eq: 2x2 unit square}, there are only 6 qubits remaining, and the ratio between qubits and fermions is $\frac{6}{4}=1.5$.

By the Clifford circuits in Fig.~\ref{fig:square_lattice}, we eliminate stabilizers on odd faces and convert the stabilizers on even faces to toric-code-like stabilizers. Next, we analyze the logical operators which represent the hopping of fermion after the conjugation. Here the convention of the fermionic hopping is $S_e\equiv i\gamma_{L(e)}\gamma_{R(e)}'$.
There are four types of fermionic hopping operator (after removing the degrees of freedom in Eq.~\eqref{eq: r=1_5 stabilizer cancel})
\begin{eqs}\label{eq:r=1.5_hopping}
    \begin{gathered}
   \xymatrix@=1cm{%
    {}\ar@{-}[r]|{}&  \\
    {}\ar@{-}[r]|{\displaystyle Z} \ar@{-}[u]|{}&{}\ar@{-}[u]|{\displaystyle X_e} \ar@{}[lu]|{\displaystyle \text{odd}}  }
    \end{gathered}
     &\begin{gathered}
    \xymatrix{
        {}\ar[r] &{}}
    \end{gathered} \quad
    \begin{gathered}
   \xymatrix@=1cm{%
    {}\ar@{-}[r]|{}& {}\ar@{-}[r]|{\displaystyle Z} & \\
    {}\ar@{-}[r]|{\displaystyle Z} \ar@{-}[u]|{}&{}\ar@{-}[u]|{\displaystyle X_e} \ar@{}[lu]|{\displaystyle \text{odd}}  &}
    \end{gathered},\\[10pt]
    \begin{gathered}
   \xymatrix@=1cm{%
    {}\ar@{-}[r]|{}&  \\
    {}\ar@{-}[r]|{\displaystyle Z} \ar@{-}[u]|{}&{}\ar@{-}[u]|{\displaystyle X_e} \ar@{}[lu]|{\displaystyle \text{even}}  }
    \end{gathered}
     &\begin{gathered}
    \xymatrix{
        {}\ar[r] &{}}
    \end{gathered} \quad
     \begin{gathered}
   \xymatrix@=1cm{%
    {}\ar@{-}[r]|{}& {}\ar@{-}[r]|{\displaystyle Z} &{}\ar@{-}[d]|{\displaystyle Z} \\
    {}\ar@{-}[r]|{\displaystyle Y} \ar@{-}[u]|{}&{}\ar@{-}[u]|{\displaystyle e} \ar@{}[lu]|{\displaystyle \text{even}}  \ar@{-}[r]|{\displaystyle X} \ar@{-}[d]|{\displaystyle Y} &\\&&}
    \end{gathered},\\[10pt]
    \begin{gathered}
   \xymatrix@=1cm{%
    {}\ar@{-}[r]|{\displaystyle X_e}&  \\
    {}\ar@{-}[u]|{\displaystyle Z} \ar@{-}[r]|{}&{}\ar@{-}[u]|{} \ar@{}[lu]|{\displaystyle \text{odd}}}
    \end{gathered}
     &\begin{gathered}
    \xymatrix{
        {}\ar[r] &{}}
    \end{gathered} \quad
    \begin{gathered}
   \xymatrix@=1cm{%
    {}\ar@{-}[r]|{\displaystyle X_e}&  \\
    {}\ar@{-}[u]|{} \ar@{-}[r]|{}&{}\ar@{-}[u]|{} \ar@{}[lu]|{\displaystyle \text{odd}}}
    \end{gathered},\\[10pt]
    \begin{gathered}
   \xymatrix@=1cm{%
    {}\ar@{-}[r]|{\displaystyle X_e}&  \\
    {}\ar@{-}[u]|{\displaystyle Z} \ar@{-}[r]|{}&{}\ar@{-}[u]|{} \ar@{}[lu]|{\displaystyle \text{even}}}
    \end{gathered}
     &\begin{gathered}
    \xymatrix{
        {}\ar[r] &{}}
    \end{gathered} \quad
    \begin{gathered}
   \xymatrix@=1cm{%
    {}\ar@{-}[r]|{\displaystyle X_e}&  \\
    {}\ar@{-}[u]|{} \ar@{-}[r]|{}&{}\ar@{-}[u]|{} \ar@{}[lu]|{\displaystyle \text{even}}}
    \end{gathered},
\end{eqs}
and two types of flux operators
\begin{eqs}\label{eq:r=1.5_flux}
    \begin{gathered}
   \xymatrix@=1cm{%
    {}\ar@{-}[r]|{\displaystyle Z}&  \\
    {}\ar@{-}[u]|{\displaystyle Z} \ar@{-}[r]|{\displaystyle Z}&{}\ar@{-}[u]|{\displaystyle Z} \ar@{}[lu]|{\displaystyle \text{odd}}}
    \end{gathered}
     &\begin{gathered}
    \xymatrix{
        {}\ar[r] &{}}
    \end{gathered} \quad
    \begin{gathered}
   \xymatrix@=1cm{%
    {}\ar@{-}[r]|{\displaystyle Z}&  \\
    {}\ar@{-}[u]|{} \ar@{-}[r]|{\displaystyle Z}&{}\ar@{-}[u]|{\displaystyle Z} \ar@{}[lu]|{\displaystyle \text{odd}}}
    \end{gathered},\\[10pt]
    \begin{gathered}
   \xymatrix@=1cm{%
    {}\ar@{-}[r]|{\displaystyle Z}&  \\
    {}\ar@{-}[u]|{\displaystyle Z} \ar@{-}[r]|{\displaystyle Z}&{}\ar@{-}[u]|{\displaystyle Z} \ar@{}[lu]|{\displaystyle \text{even}}}
    \end{gathered}
     &\begin{gathered}
    \xymatrix{
        {}\ar[r] &{}}
    \end{gathered} \quad
    \begin{gathered}
   \xymatrix@=1cm{%
    {}\ar@{-}[r]|{\displaystyle Z}&  \\
    {}\ar@{-}[u]|{\displaystyle Z} \ar@{-}[r]|{\displaystyle Z}&{}\ar@{-}[u]|{} \ar@{}[lu]|{\displaystyle \text{even}}}
    \end{gathered}.
\end{eqs}
We note that the stabilizer in Eq.~\eqref{eq: r=1_5 stabilizer} is the same as the stabilizer of the compact encoding in Ref.~\cite{compact_fermion} (up to the relabeling of Pauli matrices $X$, $Y$, $Z$). Since the stabilizers are the same, the space of logical operators must be equivalent. We can redefine the bottom two lines of Eq.~\ref{eq:r=1.5_hopping} as ``fermion parity'' by re-pairing of Majorana fermions as Fig.~\ref{fig:compact_regrouping}, and reproduce the compact encoding in Ref.~\cite{compact_fermion}.

\begin{figure}[t]
    \centering
    \includegraphics[width=0.45\textwidth]{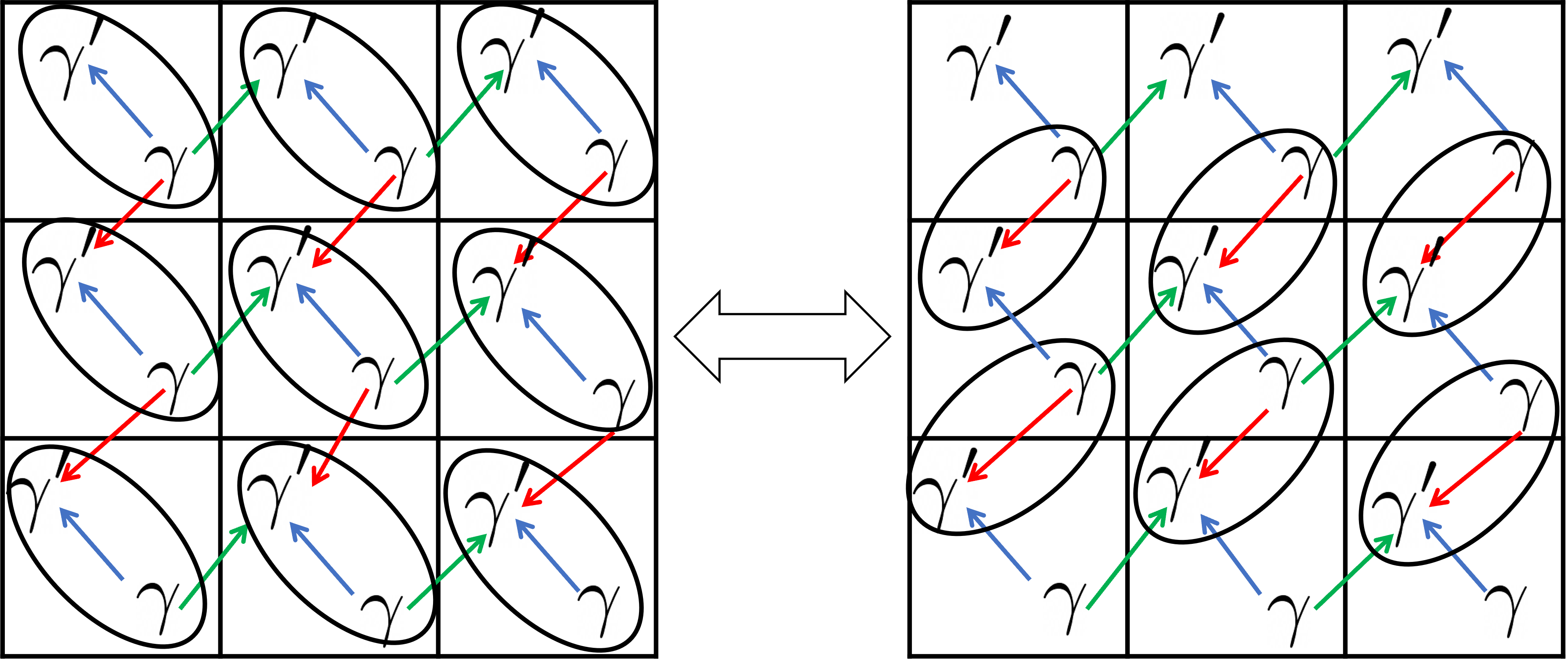}
    \caption{The $r=1.5$ construction is the same as the compact fermion-to-qubit mapping \cite{compact_fermion} after the re-paring of Majorana fermions above. Each circle represents a complex fermion generated by the two Majorana fermions. The underlying arrows specify the order to form a fermion. This is a Kasteleyn orientation ensuring the re-pairing is well-defined \cite{TF16}.}
    \label{fig:compact_regrouping}
\end{figure}

\subsection{Super-compact fermion-to-qubit mapping with ratio \texorpdfstring{$r=1.25$}{}} \label{sec:square_r_1.25}

\begin{figure}[h]
    \centering
    \includegraphics[width=0.3\textwidth]{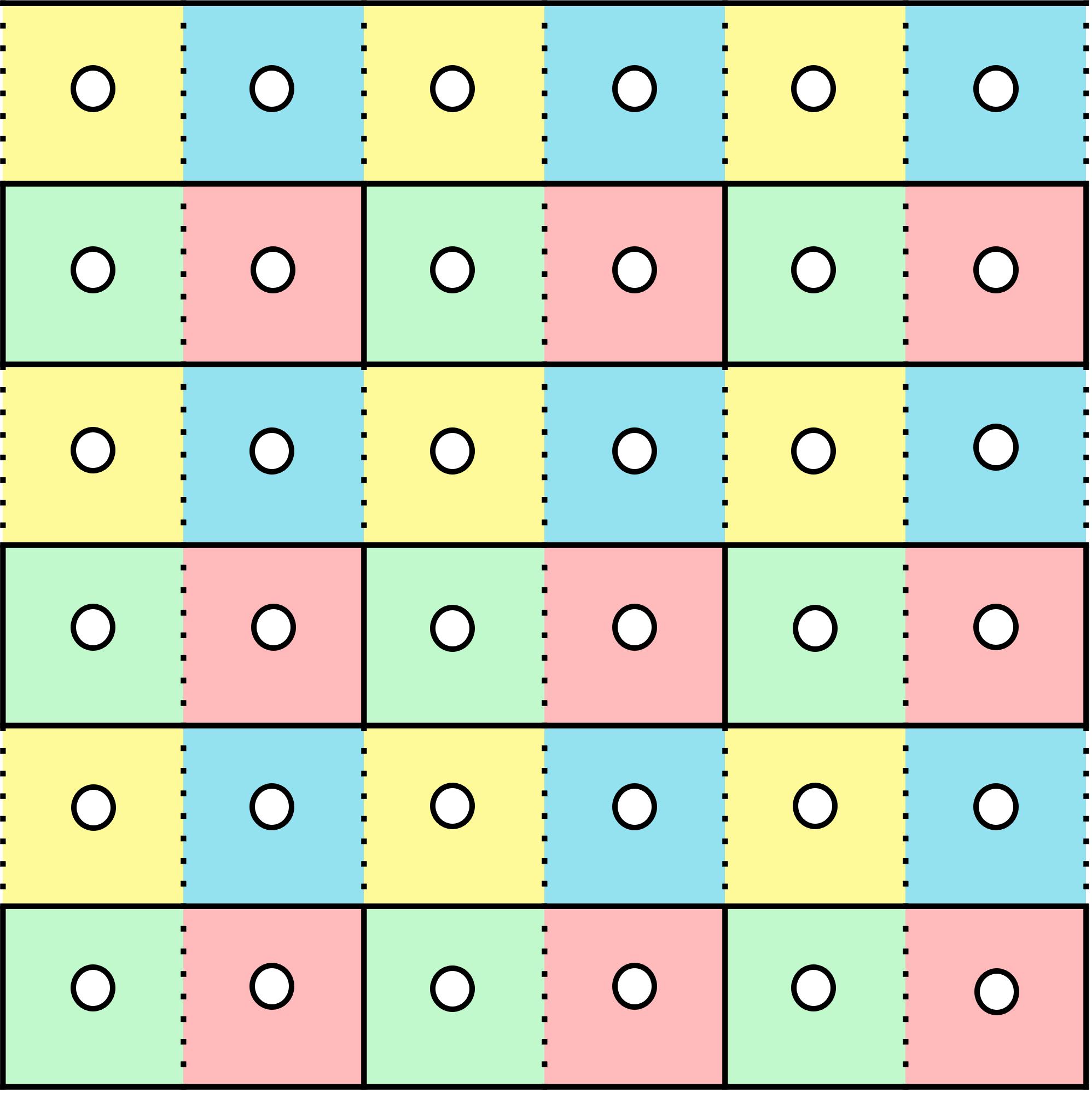}
    \caption{Square lattice, yellow squares are in class 1; blue squares are in class 2; red squares are in class 3; green squares are in class 4. Dots are fermionic modes encoded inside squares. Each solid line has one qubit. There are no qubits on dashed lines.}
    \label{}
\end{figure}

Based on the $r=1.5$ construction in the previous section, which is obtained from conjugating the original 2d bosonization by the Clifford circuit shown in Fig.~\ref{fig:square_lattice}, we further conjugate it by the Clifford circuit in Fig.~\ref{fig:square_lattice_step3}.  In the $r=1.5$ construction, we label squares by ``even" and ``odd". Since the translational invariant Clifford circuit in Fig.~\ref{fig:square_lattice_step3} acts on a $2\times 2$ cell, we color squares by 4 different colors: yellow, blue, red, and green and call them class 1, 2, 3, and 4 squares. Classes 1 and 3 belong to ``odd" faces, while classes 2 and 4 belong to the ``even" faces. 

\begin{figure*}[h]
    \centering
    \includegraphics[width=0.8\textwidth]{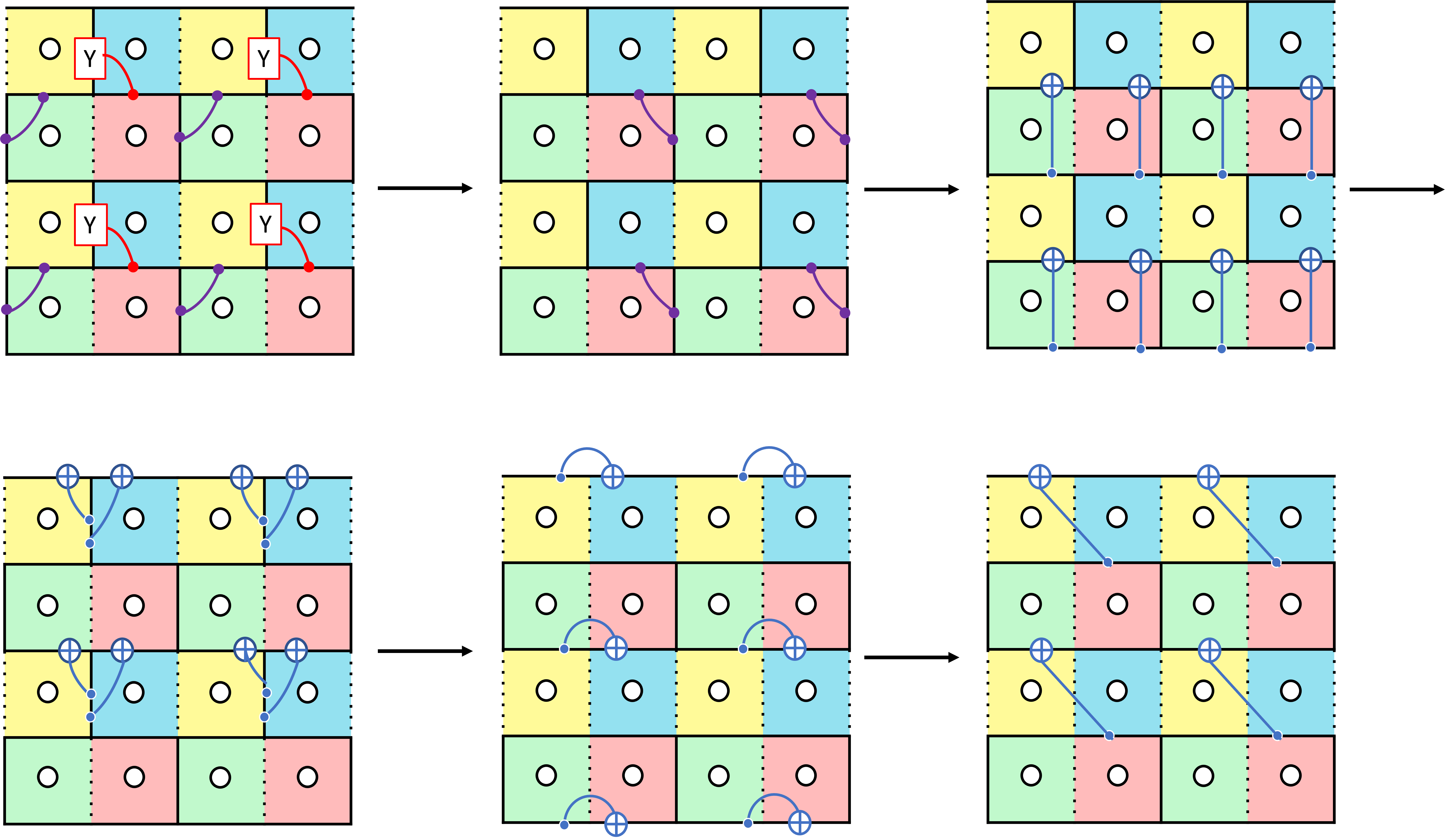}
    \caption{The finite-depth Clifford circuit to construct bosonization with $r=1.25$.}
    \label{fig:square_lattice_step3}
\end{figure*} 

\begin{figure*}[h]
\centering
    \includegraphics[width=\textwidth]{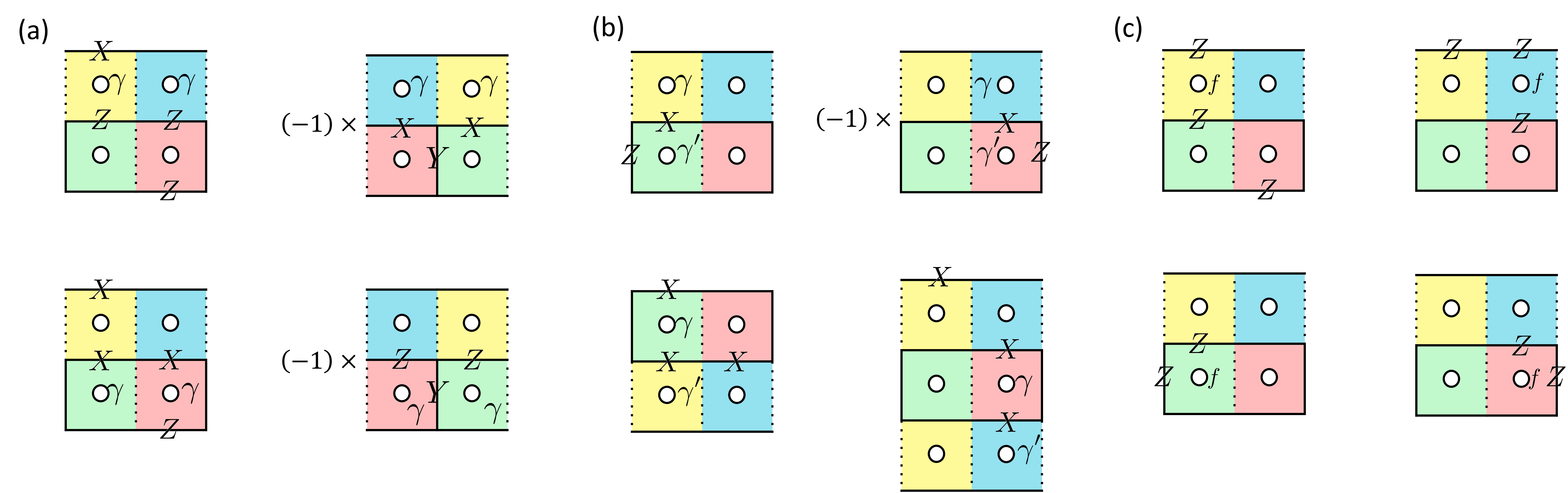}
    \caption{(a) represents the nearest-neighbor horizontal hopping terms; (b) represents the near-neighbor vertical hopping terms; (c) represents the fermion parity operators}
    \label{fig:r1.25_logical}
\end{figure*}

\begin{figure*}[h]
    \centering
    \includegraphics[width=0.5\textwidth]{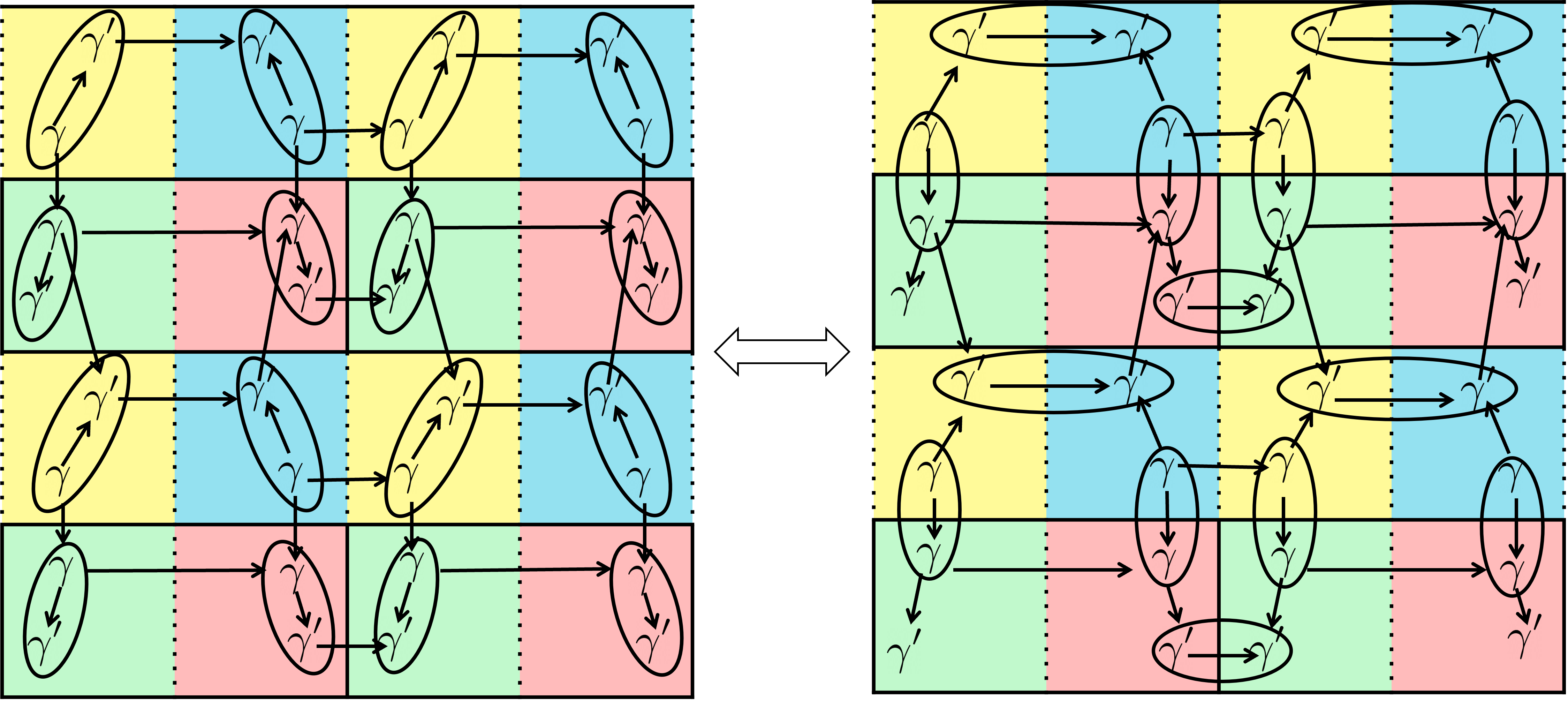}
    \caption{The super-compact fermion-to-qubit mapping is the $r=1.25$ construction after the fermion re-pairing. The arrows specify the order to form a complex fermion by two Majorana fermions in a circle. The arrows form a Kasteleyn orientation \cite{Kasteleyn_orientation_spin_structure}.}
    \label{fig:super_compact_regroup}
\end{figure*}

Conjugated by the circuits in Fig.~\ref{fig:square_lattice_step3}, the stabilizers enveloping blue and yellow squares become

\begin{equation*}\label{eq:r_1.25_stb_step1}
    (-1)\times \begin{gathered}
   \includegraphics[width=0.14\textwidth]{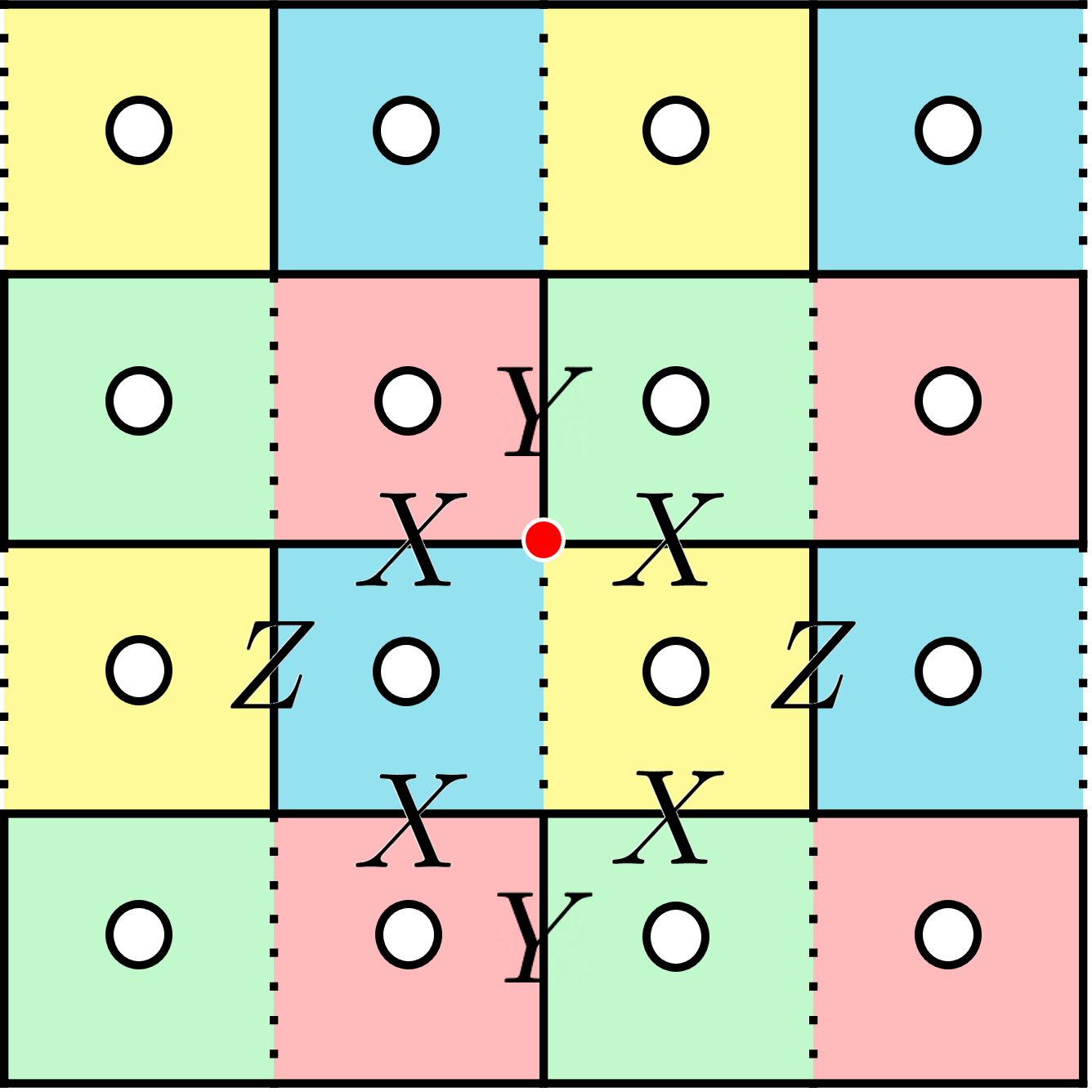}
    \end{gathered} 
    \begin{gathered}
    \xymatrix{
        {}\ar[r] &{}}
    \end{gathered} 
    (-1)\times \begin{gathered}
   \includegraphics[width=0.14\textwidth]{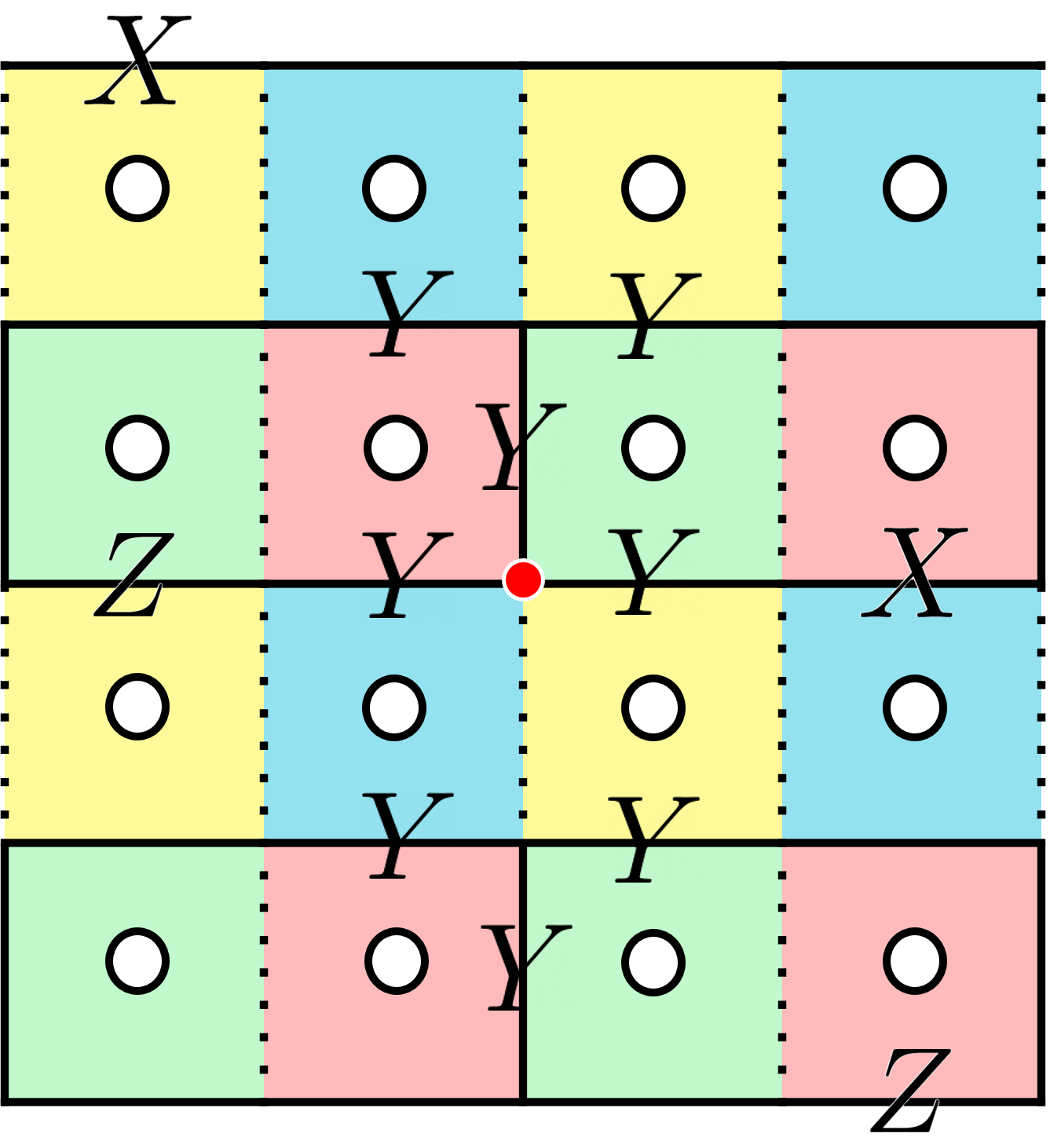}
    \end{gathered}.
\end{equation*}
However, the stabilizers enveloping green and red squares become
\begin{equation*}\label{eq:r_1.25_trivial_stabilizer}
    (-1)\times \begin{gathered}
   \includegraphics[width=0.14\textwidth]{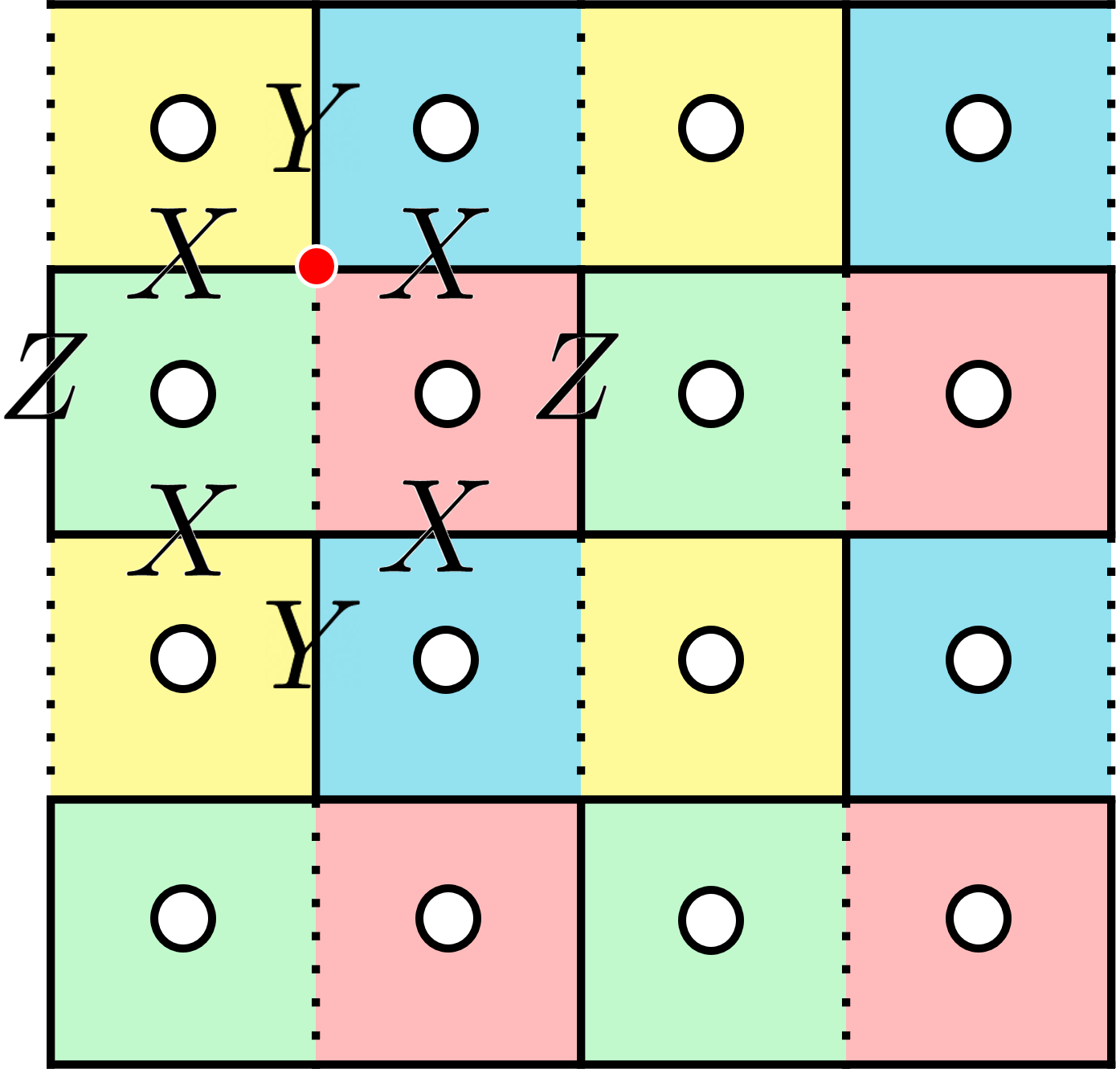}
    \end{gathered} 
    \begin{gathered}
    \xymatrix{
        {}\ar[r] &{}}
    \end{gathered} 
    (-1)\times \begin{gathered}
   \includegraphics[width=0.14\textwidth]{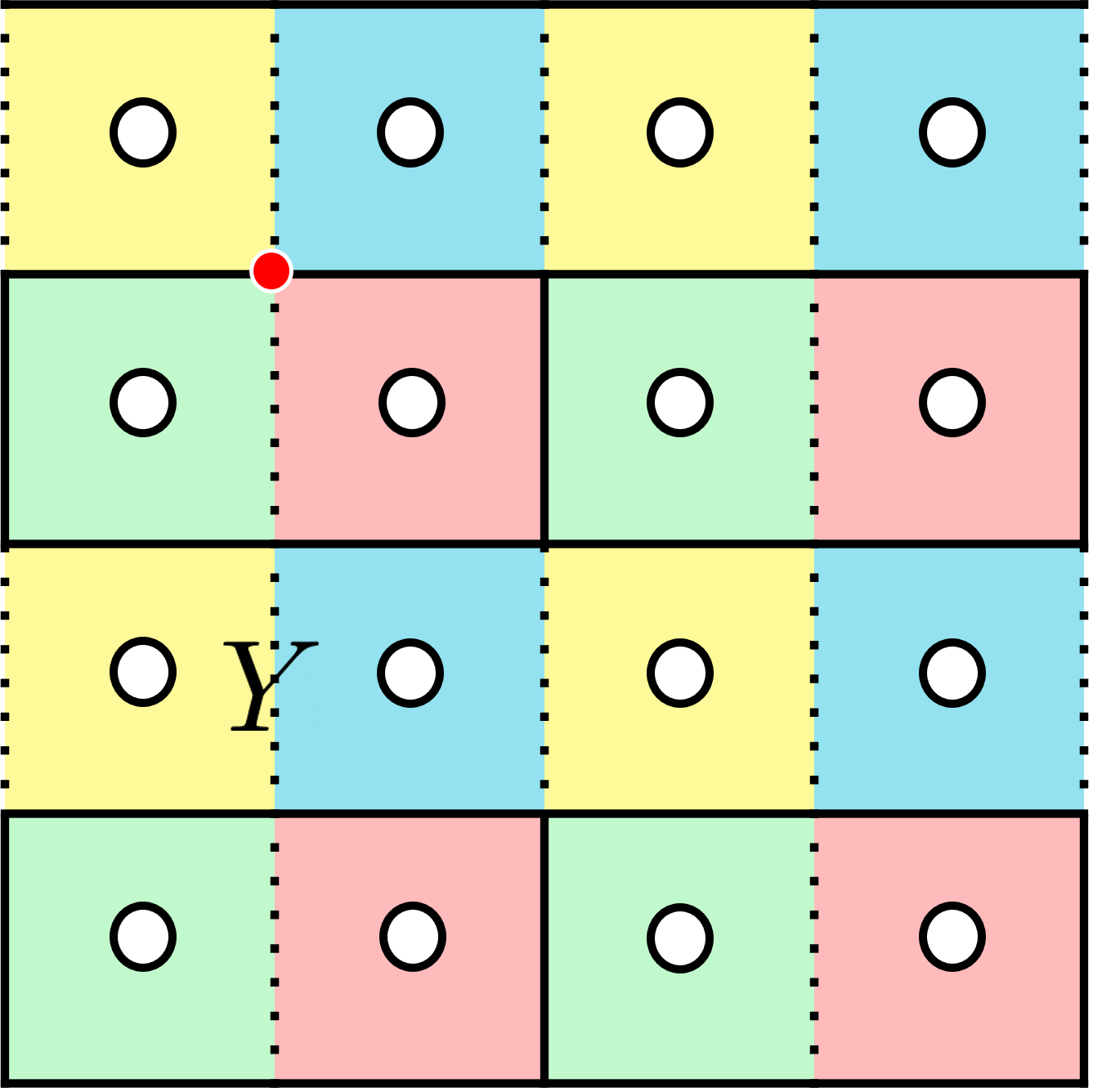}
    \end{gathered},
\end{equation*}
which is trivial. Then, we can simply remove the qubits on the boundaries between yellow and blue squares, which reduces the qubit-fermion ratio to $r=1.25$.

Similarly, we conjugate Eq.~\eqref{eq:r=1.5_hopping} and \eqref{eq:r=1.5_flux} by the Clifford circuit in Fig.~\ref{fig:square_lattice_step3}, and the result of these logical operators are listed in Fig.~\ref{fig:r1.25_logical}. This gives the super-compact fermion-to-qubit mapping demonstrated in Sec.~\ref{sec:super_compact} by a re-pairing of Majorana fermions (Fig.~\ref{fig:super_compact_regroup}) and a slight lattice deformation.

\subsection{General Construction for compact fermion-to-qubit mappings}\label{sec:Clifford_construction}


In this section, we describe a general method to construct fermion-to-qubit mappings with a reduced qubit-fermion ratio from the exact bosonization. The exact bosonization contains gauge constraints (stabilizers) Eq.~\eqref{eq:gauge constraint at vertex} supported on faces $f$ (northeast to vertices $v$), and we rename $G_v$ as $G_f$ for convenience. We enlarge the unit cell and will show that it is always possible to apply finite-depth gLU operators such that a portion of stabilizers can be mapped to a single Pauli matrix. More precisely, we are going to prove that the stabilizer on each white face below can be mapped to a single Pauli matrix:
\begin{equation}
    \includegraphics[width=0.3\textwidth]{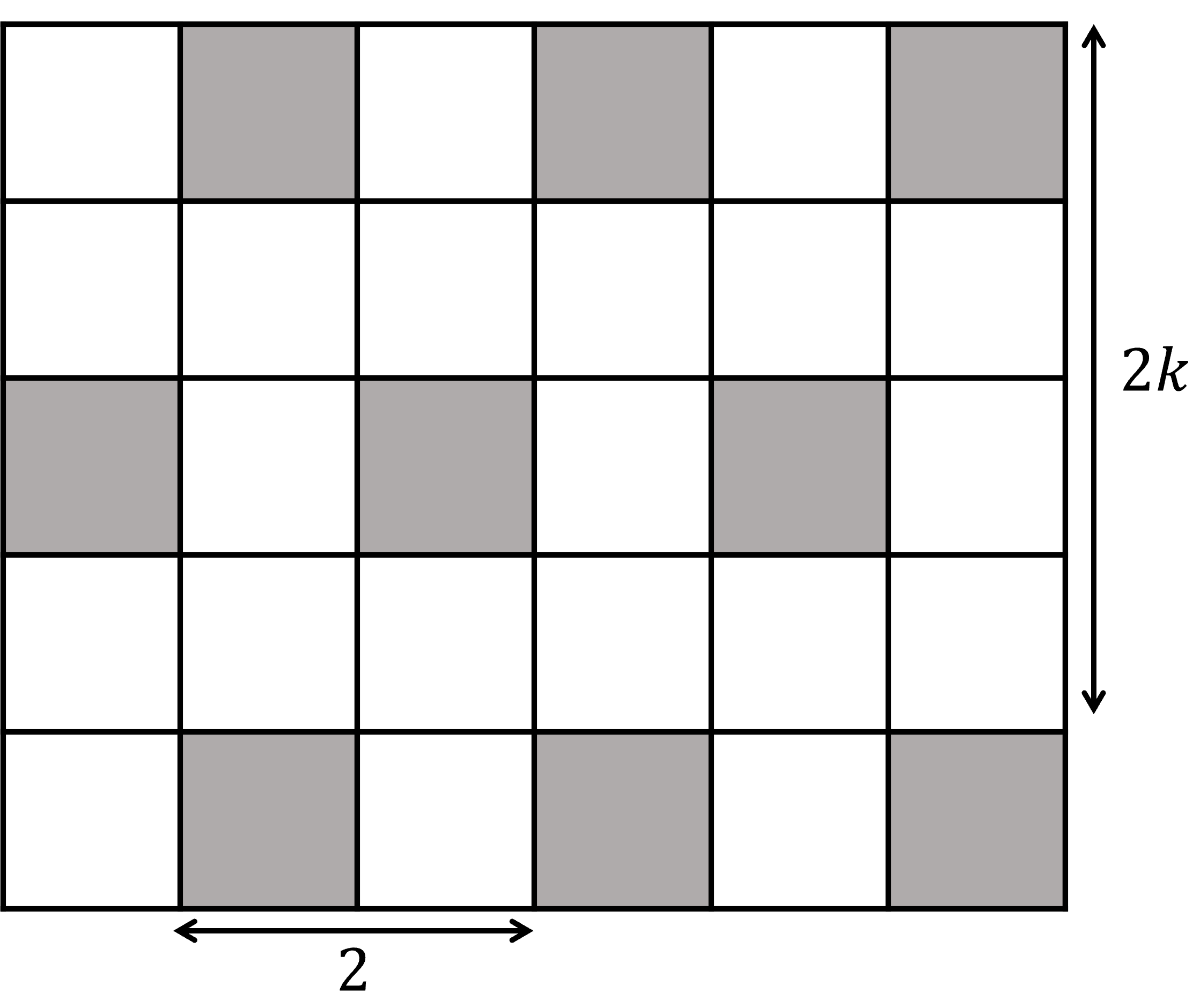},
\end{equation}
where $k$ can be any positive integer.\footnote{The portion of grey faces over all faces is $\frac{1}{2k}$. After removing stabilizers on white faces, the qubit-fermion ratio becomes $r=1+\frac{1}{2k}$.}

Instead of transforming $G_f$ on white faces directly, we are going to prove a stronger statement: the gauge constraints $G_f$ (Eq.~\eqref{eq:gauge constraint at vertex}) on white faces, the hopping operators $U_e$ (Eq.~\eqref{eq: Ue on horizontal edge}) across horizontal edges, and the operators 
\begin{eqs}
    G'_f \equiv
    \begin{gathered}
        \xymatrix@=1cm{%
        \ar@{}[dr]|{\mathlarger f} \ar@{-}[r]|{} \ar@{-}[d]|{}&{}\ar@{-}[r]|{\displaystyle Y} & {}\ar@{-}[d]|{\displaystyle Z} \\
        \ar@{-}[r]|{} &{}\ar@{-}[u]|{\displaystyle X} &
        }
    \end{gathered},
\end{eqs}
on grey faces, can all be mapped to a single Pauli matrix simultaneously under a finite-depth gLU circuit. These operators on the square lattice are shown as
\begin{eqs}
    \includegraphics[width=0.35\textwidth]{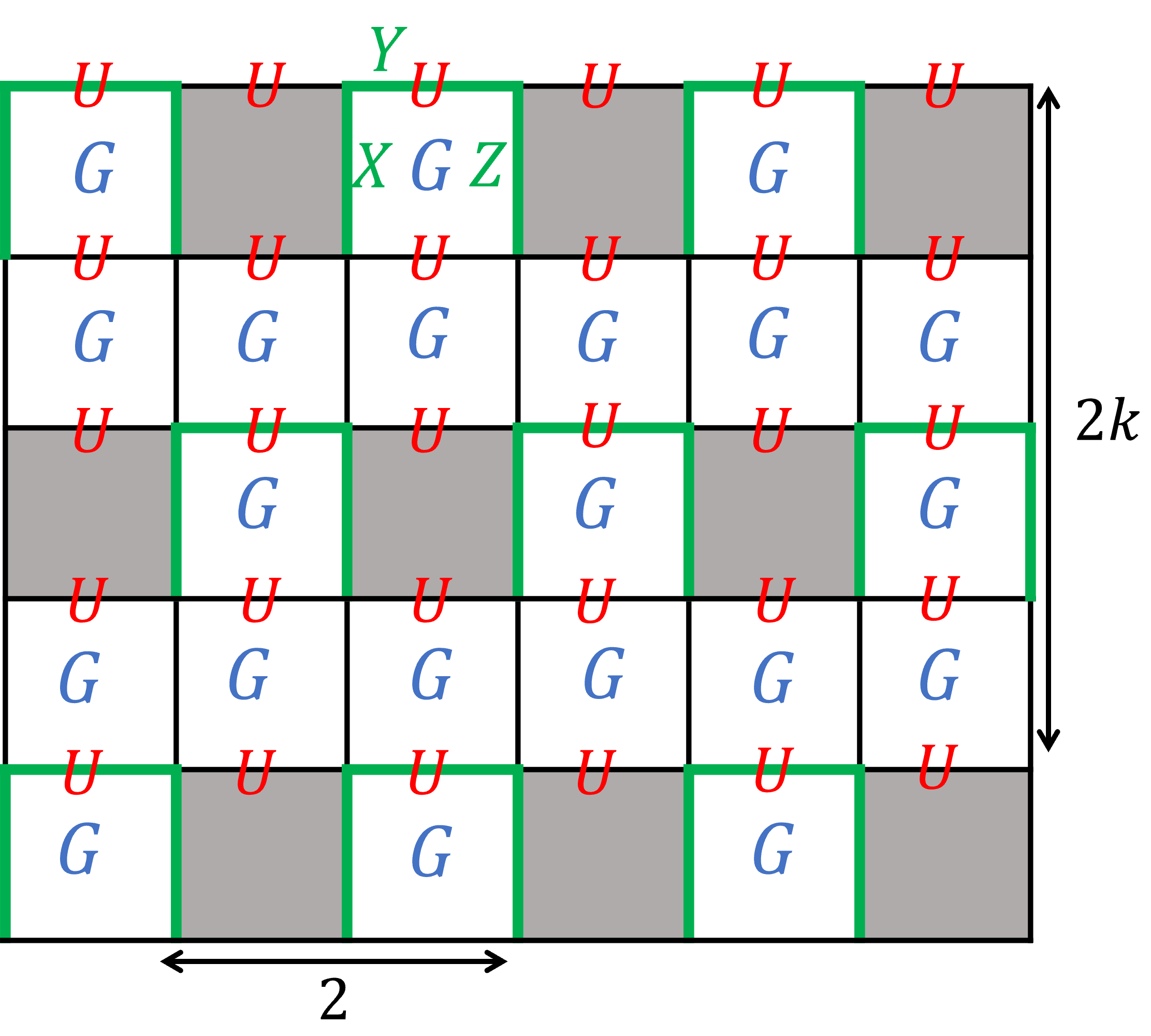}.
\label{eq: separators}
\end{eqs}
To prove the above statement, we need to introduce a lemma:
\begin{lemma}\label{lemma: QCA to circuit}
    Given $\tilde Z_e$ and $\tilde X_e$ for all edges that are products of Pauli matrices on a neighborhood of the edge $e$ satisfying the Pauli algebra,
    \begin{equation*}
        [\tilde X_e, \tilde X_{e'}]=[\tilde Z_e, \tilde Z_{e'}]=0, \quad \tilde X_e \tilde Z_{e'} = (-1)^{\delta_{e,e'}} \tilde Z_{e'} \tilde X_e,
    \end{equation*}
    there exists a finite-depth gLU transformation mapping $\tilde X_e, \tilde Z_e$ to $X_e, Z_e$ (a single Pauli on edge $e$).
\end{lemma}
\begin{proof}
    The (Clifford) quantum cellular automata (QCA) in two spatial dimensions are simply (Clifford) local unitary circuits and shifts \cite{FH20, Haah_Clifford_QCA21}. The map $\alpha$
    \begin{eqs}
        \alpha(X_e) = \tilde X_e,
        \quad
        \alpha(Z_e) = \tilde Z_e,
    \label{eq: definition of QCA}
    \end{eqs}
    defines a QCA and therefore can be decomposed to a Clifford circuit and shifts. For the shift operator, we can introduce ancilla in the $|0\rangle$ states and define the shift operator moving the ancilla in the opposite direction, such that the net flow of qubits is zero. Then, this shift operator can be expressed by a local unitary circuit (involving the ancilla degrees of freedom). In the end, these ancilla are still in the $|0\rangle$ states and can be removed by a finite-depth gLU transformation. Therefore, there exists a finite-depth gLU transformation from $X_e, Z_e$ to $\tilde X_e, \tilde Z_e$ and vice versa.
\end{proof}
\begin{lemma}
    Given operators $\tilde Z_e$ (separators) and $\bar X_e$ (flippers) that are products of Pauli matrices on a neighborhood of the edge $e$ satisfying
    \begin{equation}
        [\tilde Z_e, \tilde Z_{e'}]=0, \quad \bar X_e \tilde Z_{e'} = (-1)^{\delta_{e,e'}} \tilde Z_{e'} \bar X_e,
    \label{eq: flipper algebra}
    \end{equation}
    there exist operators $\tilde X_e$ that are products of Pauli matrices on a neighborhood of edges $e$ such that
    \begin{equation*}
        [\tilde X_e, \tilde X_{e'}]=[\tilde Z_e, \tilde Z_{e'}]=0, \quad \tilde X_e \tilde Z_{e'} = (-1)^{\delta_{e,e'}} \tilde Z_{e'} \tilde X_e.
    \end{equation*}
    In other words, if the flippers do not commute with themselves, they can be modified such that the Pauli algebra is satisfied.
\end{lemma}
\begin{proof}
    If $\bar X_e$ and $\bar X_{e'}$ do not commute,
    \begin{eqs}
        \bar X_e \bar X_{e'} = - \bar X_{e'} \bar X_e,
    \end{eqs}
    we define
    \begin{eqs}
        \tilde X_e \equiv \bar X_e \tilde Z_{e'}, \quad \tilde X_{e'}\equiv \bar X_{e'}.
    \label{eq: redefine X}
    \end{eqs}
    Notice that $\tilde Z_{e'}$ only affects the commutation relation between $e$ and $e'$ and this fixes the commutation for the $X$ part and leaves $Z$ part unchanged. Therefore, $\tilde X_e$ and $\tilde Z_e$ satisfy the Pauli algebra.
\end{proof}
The operators $\tilde{Z}_e$ and $\bar X_e$ are call separators and flippers \cite{FHH18}. Once the separators and flippers are given, a QCA is defined by Eq.~\eqref{eq: definition of QCA} (after defining $\tilde X_e$ by Eq.~\eqref{eq: redefine X}). By lemma~\ref{lemma: QCA to circuit}, the separator can be mapped to a single Pauli matrix by a finite-depth gLU transformation.

\begin{figure}[h]
    \centering
    \includegraphics[width=0.4\textwidth]{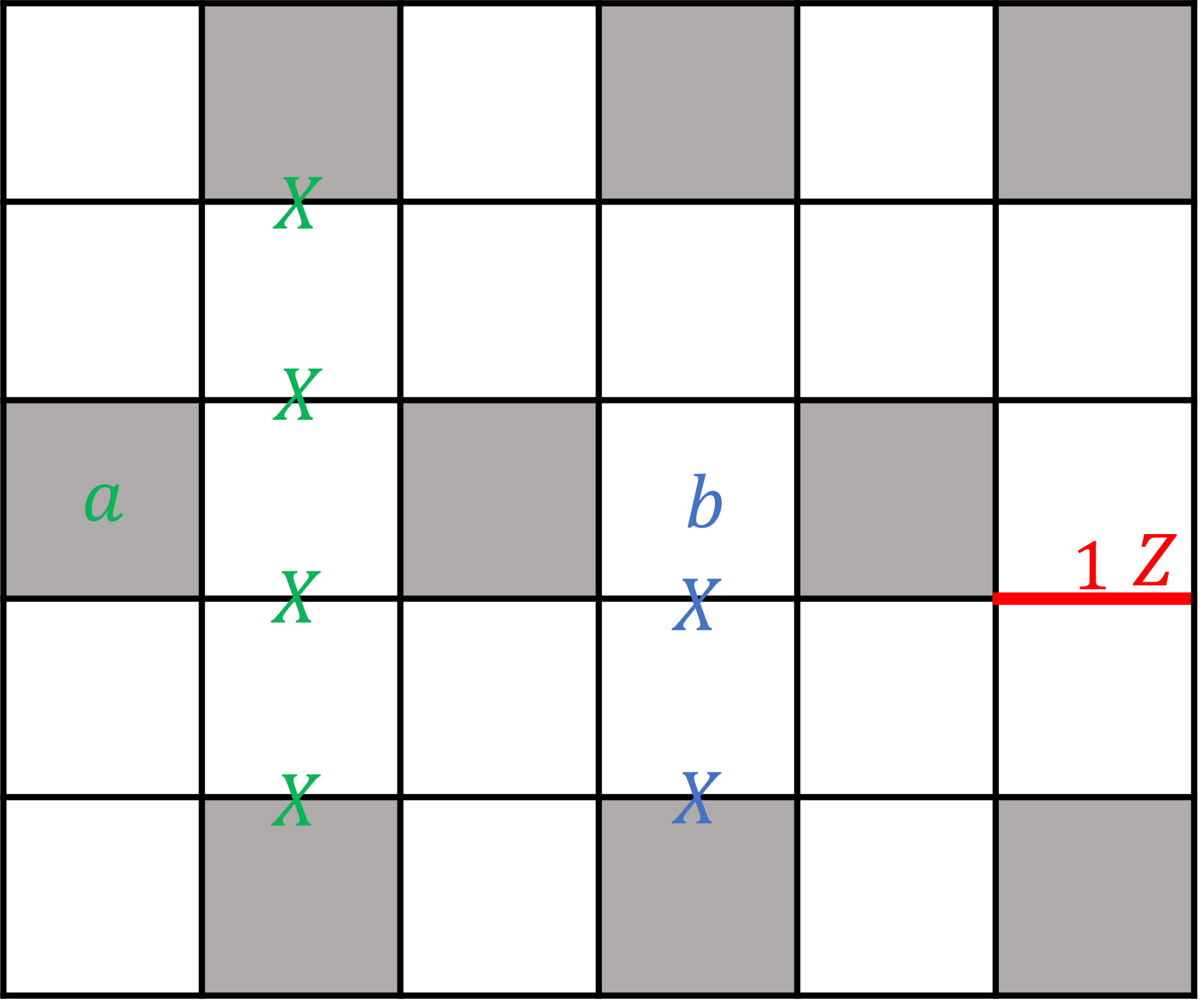}
    \caption{The (potential) flippers. For $G'_a$ on the grey face $a$, its flipper is the product of $X$ connecting two grey faces on its right column, shown by the green operator. For $G_b$ on the white face, its potential flipper is the product of $X$ connecting to a grey face below, shown by the blue operator. This potential flipper may anti-commute with $G'$ on a grey face, which can be fixed by attaching the flipper for this $G'$. For $U_{e_1}$ on a horizontal edge $e_1$, the potential flipper is $Z_{e_1}$, which flips exactly one $U_e$ and anti-commute with some $G_f$ and $G'_f$ on white and grey faces. This can be fixed by attaching the flippers for these $G_f$ and $G'_f$ to the potential flipper of $U_{e_1}$.}
    \label{fig:flippers}
\end{figure}

The operators $G_f$ on white faces, $U_e$ on horizontal edges, and $G'_f$ in grey faces in \eqref{eq: separators} are the separators $\tilde Z_e$.
Now, we are going to describe their flippers:
\begin{enumerate}
    \item For $G_f'$ on grey faces, we define its flippers by the product of $X_e$ ($m$-strings of the toric code) connecting two grey faces on the column to the right, as shown in Fig.~\ref{fig:flippers}. It can be checked that this $m$-string only violates exactly one $G'_f$ and commute with all other separators $G_f$ and $U_e$.
    \item A potential flipper\footnote{The potential flipper is an operator satisfying the algebra \eqref{eq: flipper algebra} partially. For example, it may anti-commute with extra separators $\tilde Z_{e'}$. This issue can be fixed by attaching other operators to this potential flipper.} for the separator $G_f$ on a white face is the product of $X$ connecting the white face to the grey face below (Fig.~\ref{fig:flippers}). This operator flips exactly one $G_f$ on white faces and commutes with $U_e$, but it may anti-commute with a $G'_f$ on a grey face. In this case, we can always attach the flipper for this $G'_f$ (found in step~1) to the potential flipper. This operator becomes the true flipper for a single $G_f$.
    \item For $U_e$ on a horizontal edge, we start with a potential flipper $Z$ on this edge $e$. It is obvious that it flips only one $U_e$ and may anti-commute a finite number of $G_f$ and $G_f'$ on white and grey faces. Since we have already found the flippers for $G_f$ and $G_f'$, we can attach these flippers to the potential flipper such that the combined operator commutes will all separators except one $U_e$.
\end{enumerate}

We have found the complete set of separators and flippers on the square lattice. By Lemma~\ref{lemma: QCA to circuit}, the $G_f$ on each white face can be mapped to a single Pauli matrix.



\section{Equivalence between fermion-to-qubit mappings and the exact bosonization}\label{sec:equivalence_relation}

In this section, we argue that any locality preserving fermion-to-qubit mappings\footnote{To be precise, we consider the mapping between local fermionic observables and local products of Pauli matrices.} in two spatial dimensions can be connected to the exact bosonization by a finite-depth gLU transformation. First, given a fermion-to-qubit mapping, it must contain the flux operators $\tilde W$ (images of the local fermion parity) and the gauge constraints $\tilde G$ (images of the product of fermionic hopping terms in a small closed loop). On a torus, we can define a Pauli stabilizer code as
\begin{eqs}
    H = - \sum \tilde G - \sum \tilde W.
\label{eq: stabilizer code}
\end{eqs}
Over two large cycles of the torus, we have the $4$-fold ground state degeneracy since we do not impose the fermionic constraints on the large cycles. The code distance is linear in the system size since the logical operator is the product of hopping along with the large cycles. It is proven in Ref.~\cite{Haah_Zp_21} that any translationally invariant $Z_p$ Pauli stabilizer model with a linear code distance is decomposed by a local Clifford circuit of constant depth into a finite number of copies of the toric code for any prime $p$.\footnote{In this paper, we only work on qubits, which corresponds to $p=2$. Therefore, the theorem in Ref.~\cite{Haah_Zp_21} is valid.} Since the degeneracy is 4 on the torus, the above stabilizer code Eq.~\eqref{eq: stabilizer code} must be a single copy of toric code up to a Clifford circuit. Therefore, $\tilde G$ and $\tilde W$ are related to $G_v$ and $W_f$ in the exact bosonization in Sec.~\ref{sec: review of 2d bosonization} by a gLU transformation (since the toric code defined on different lattice should be related by gLU transformation to add or remove qubits). 

In the following part of this section, we will explicitly demonstrate how to transform many well-known fermion-to-qubit mappings in literature to the exact bosonization.

\subsection{Bravyi-Kitaev superfast simulation}\label{sec:Bravyi_Kitaev}

The Bravyi-Kitaev superfast simulation (BKSF) in Ref.~\cite{bravyi_kitaev} is a method to encode fermionic operators into Pauli operators. BKSF encodes complex fermions at vertices $v$ by qubits on edges $e$. The key idea of BKSF is to assign an arbitrary ordering of edges around each vertex and write down the logical operators according to the ordering. For vertex $v$, we label the edges connected to $v$ by $(v,i)$,  $i=1,2,3,4$ on 2d square lattice, shown in Fig.~\ref{fig:BKSF_labeling}.

The logical operators $\tilde{A}^{\text{BK}}_e$ and $\tilde{B}^{\text{BK}}_v$ are defined as
\begin{eqs}
    &\tilde{A}^{\text{BK}}_e=X_e\prod_{(L(e),i)<(L(e),j)}Z_{L(e),i}\prod_{(R(e),k)<(R(e),l)}Z_{R(e),k},\\
    &\tilde{B}^{\text{BK}}_v=\prod_{e\supset v}Z_{(v,e)},
\end{eqs}
where $(L(e),j)$ is the label of edge $e$ on the vertex $L(e)$, $(R(e),l)$ is the label of edge $e$ on the vertex $R(e)$. The fermion-to-qubit mapping is
\begin{eqs}
    A_e &= i\gamma_{L(e)}\gamma_{R(e)} \longleftrightarrow \tilde A^{\text{BK}}_e, \\
    B_v &= - i \gamma_v \gamma^\prime_v \longleftrightarrow \tilde B^{\text{BK}}_v,
\end{eqs}
where $A_e$ and $B_v$ are fermionic operators defined in Sec.~\ref{sec:super_compact}. The product of $\tilde A^{\text{BK}}_e$ on any closed loop need to satisfy the condition Eq.~\eqref{eq: product of tilde Ae} due to the identity for Majorana operators

\begin{figure}
    \centering
    \includegraphics[width=0.3\textwidth]{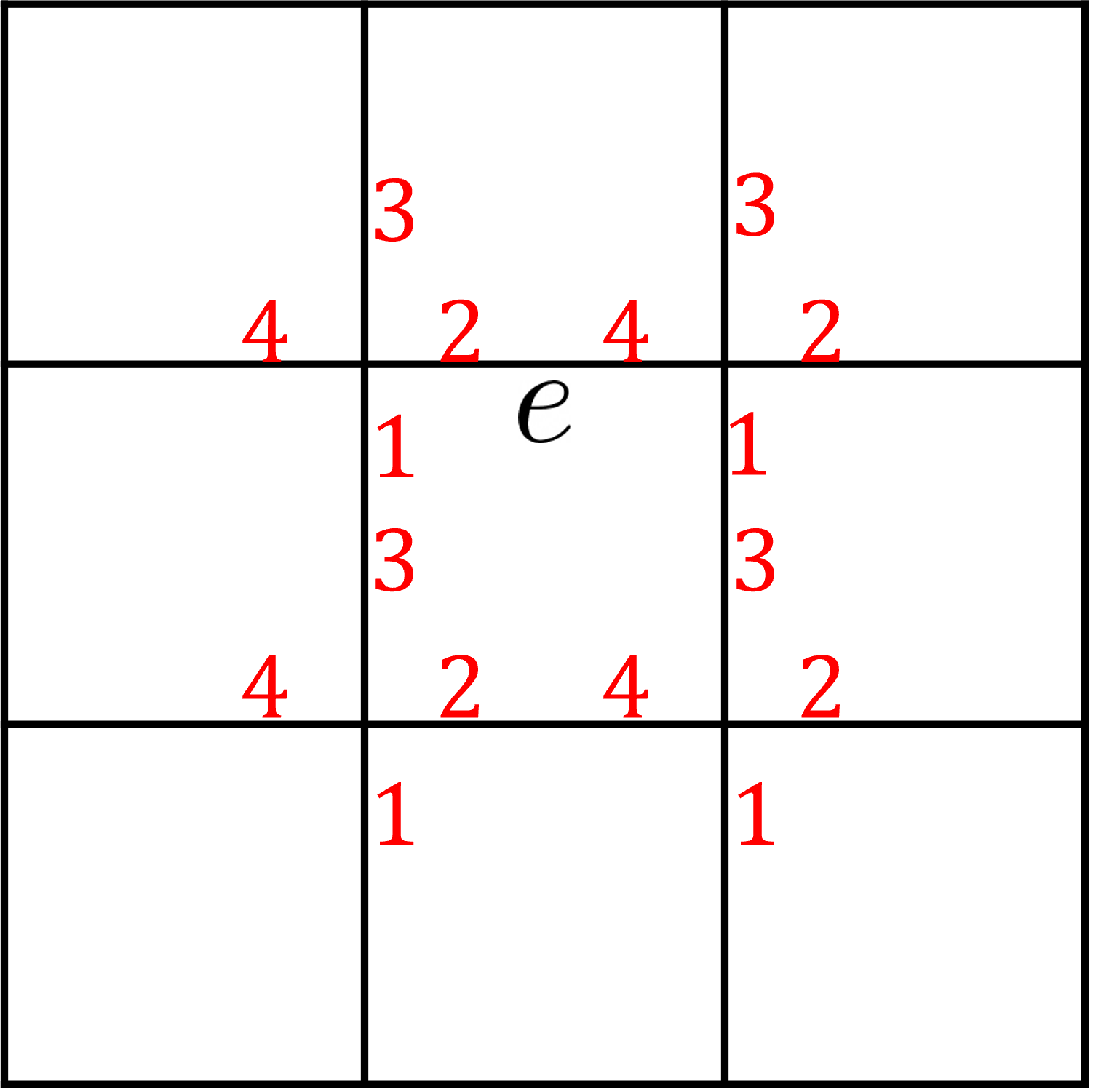}
    \caption{The ordering of edges on each vertex. The red numbers are the labels.}
    \label{fig:BKSF_labeling}
\end{figure}

By the convention in Fig.~\ref{fig:BKSF_labeling}, we have:

\begin{eqs}
    i \gamma_{L(e)} \gamma^\prime_{R(e)} \longleftrightarrow \tilde A^{\text{BK}}_e \tilde B^{\text{BK}}_{R(e)}
    =
    \begin{gathered}
    \xymatrix@=1cm{%
    &{}\ar@{-}[d]|{\displaystyle Y_e} \\{}\ar@{-}[r]|{\displaystyle Z}&
    }
    \end{gathered},\quad \begin{gathered}
    \xymatrix@=1cm{%
    {}\ar@{-}[r]|{\displaystyle Y_e}& \\{}\ar@{-}[u]|{\displaystyle Z}&
    }
    \end{gathered}.
\end{eqs}
We notice that this is the same logical operator as the exact bosonization in the dual lattice after we relabel the Pauli matrices $X$ and $Y$. The fermion parity terms in both cases are just a product of $Z$ around a vertex (a face in the dual lattice), and therefore the BKSF approach with this ordering convention is the same as the exact bosonization.



\subsection{Verstraete-Cirac auxiliary method}\label{sec:Verstraete_Cirac}
In this section, we demonstrate the equivalent relation between the Verstraete-Cirac mapping \cite{verstraete_cirac} and exact bosonization after regrouping Majorana fermions. The basic idea of the Verstraete-Cirac mapping is to eliminate the nonlocal Pauli $Z$-string from the 1d Jordan-Wigner transformation by introducing auxiliary qubits with gauge constraints. In this mapping, each site $i$ uses four Majorana modes $\gamma_i, \gamma_i', \widetilde{\gamma}_{i}, \widetilde{\gamma}_{i}'$  to encodes a complex fermion and an auxiliary complex fermion. For implementation, we put two qubits on each vertex, one for the physical complex fermion, the other for the auxiliary complex fermion. The Majorana operators $\widetilde{\gamma}_{i}$, $\widetilde{\gamma}_{i}'$ belong to the auxiliary complex fermion. The auxiliary fermions stay in the ground state of following Hamiltonian
\begin{equation}
    H_{aux}=\sum_{\{i,j\}}P_{ij}=i\sum_{\{i,j\}}\widetilde{\gamma}_{i}\widetilde{\gamma}_{j}',
\end{equation}
where the $\{i,j\}$ includes only pairs $(i,j)$ that connected by the directed edges in Fig.~\ref{fig:VC_auxiliary_ordering}. 
\begin{figure}[h]
    \centering
    \includegraphics[width=0.25\textwidth]{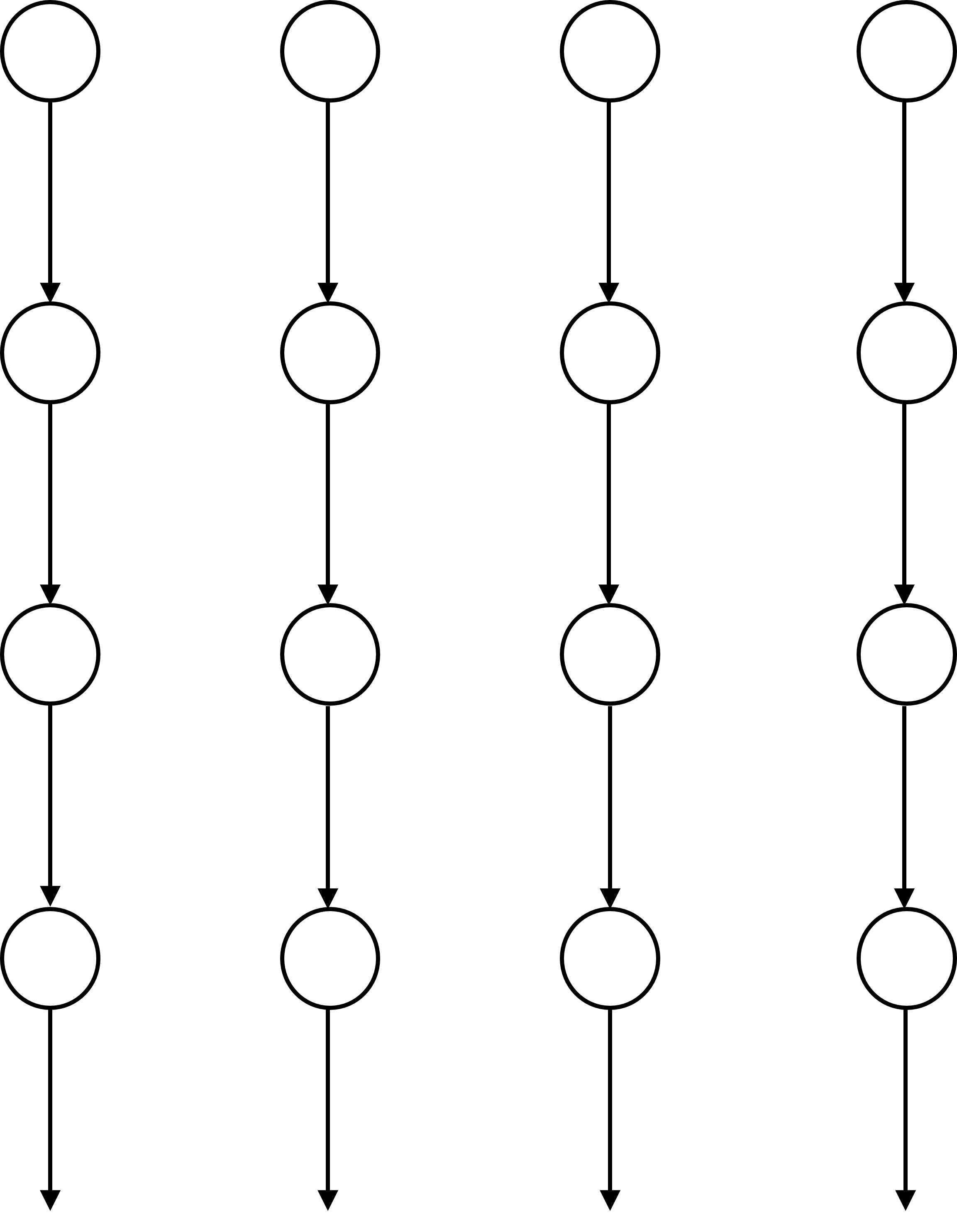}
    \caption{Graph structure of the auxiliary Hamiltonian $H_{aux}$.}
    \label{fig:VC_auxiliary_ordering}
\end{figure}
The hopping operator is modified as $c_i^\dagger c_j\rightarrow c_i^\dagger c_j(i\widetilde{\gamma}_{i} \widetilde{\gamma}_{j}')$. We apply the 1d Jordan-Wigner transformation with ordering $\{\gamma_{i_1}, \gamma'_{i_1} \} \rightarrow \{\tilde \gamma_{i_1}, \tilde \gamma'_{i_1} \} \rightarrow \{\gamma_{i_2}, \gamma'_{i_2} \} \rightarrow \{\tilde \gamma_{i_2}, \tilde \gamma'_{i_2} \} \rightarrow \cdots$, where $i_1, i_2, \cdots$ in our convention start with the first row from left to right, and move to the second and so on. However, the auxiliary Hamiltonian is a non-local Hamiltonian. To resolve this problem, we perform following substitution $P_{\{N-1,N+2\}}\rightarrow P_{\{N-1,N+2\}}P_{\{N,N+1\}}$, $P_{\{N-2,N+3\}}\rightarrow P_{\{N-2,N+3\}}P_{\{N-1,N+2\}}$ for all rows. Since all $P$ commute with each other, the auxiliary Hamiltonian $H_{aux}$ becomes local without changes in the ground state. Then the local gauge constraint (stabilizer) is $P_{ij}=\widetilde{\gamma}_i\widetilde{\gamma}_k'\widetilde{\gamma}_j\widetilde{\gamma}_l'=1$. The gauge constraint can be written as a Pauli stabilizer

\begin{eqs}
    P_{ik}=\begin{gathered}
    \xymatrix@=1cm{\ar@{-}[d]|{\displaystyle I_i}&\ar@{-}[d]|{\displaystyle Z_j}&\\
    \ar@{-}[r]|{\displaystyle \widetilde{X}_i}&\ar@{-}[r]|{\displaystyle \widetilde{Y}_j}& \\
    \ar@{-}[u]|{\displaystyle I_k} \ar@{-}[r]|{\displaystyle \widetilde{Y}_k}&\ar@{-}[r]|{\displaystyle \widetilde{X}_l} \ar@{-}[u]|{\displaystyle Z_l}&}
    \end{gathered}.
\end{eqs}
Pauli matrices $\{\widetilde{X}_n, \widetilde{Y}_n, \widetilde{Z}_n\}$ act on the auxiliary qubit $n$. We put physical qubits on the vertical edges and auxiliary qubits on the horizontal edges. Since physical qubits and auxiliary qubits are in different edges, we will not show the tilde in following text for convenience.

Its hopping operators ($S_e$ in Eq.~\eqref{eq: fermion_hopping})and fermion parity operators ($P_f$ in Eq.~\eqref{eq: fermion_parity}) are
\begin{eqs}\label{eq:VC_logical}
    &U_e=\begin{gathered}
   \xymatrix@=1cm{%
    {}\ar@{-}[d]|{\displaystyle X}
    &  \\
    {}\ar@{-}[r]|{\displaystyle Y_e} {}\ar@{-}[d]|{\displaystyle Y} & \\
    {}\ar@{-}[r]|{\displaystyle X}& 
    }
    \end{gathered}, \qquad \begin{gathered}
   \xymatrix@=1cm{%
    {}\ar@{-}[d]|{\displaystyle X}
    & {}\ar@{-}[d]|{\displaystyle X_e} \\
    {}\ar@{-}[r]|{\displaystyle Z}  &
    }
    \end{gathered} \\
    &W_f=\begin{gathered}
    \xymatrix@=1cm{%
    {}\ar@{-}[d]|{\displaystyle Z} \ar@{-}[r]|{}
    & {}\ar@{-}[d]|{} \\
    {}\ar@{-}[r]|{}  {}\ar@{}[ur]|{ f} &
    }
    \end{gathered}.
\end{eqs}

By conjugating the logical operators in Eq.~\eqref{eq:VC_logical} by the Clifford circuits shown in Fig.~\ref{fig:VC_Clifford}, the logical operators and stabilizer become

\begin{eqs}
    &U_e=\begin{gathered}
   \xymatrix@=1cm{%
    &  {}\ar@{-}[d]|{\displaystyle e}\\
    {}\ar@{-}[r]|{\displaystyle X} {}\ar@{-}[d]|{\displaystyle Z} & \\& 
    }
    \end{gathered}, \quad \begin{gathered}
   \xymatrix@=1cm{%
    &{}\ar@{-}[d]|{\displaystyle X} {}\ar@{-}[r]|{\displaystyle e}
    & \\
    {}\ar@{-}[r]|{\displaystyle Z}  &&
    }
    \end{gathered} \\
    &W_f=\begin{gathered}
   \xymatrix@=1cm{%
    &{}\ar@{-}[r]|{} & {}\ar@{-}[d]|{}\\
    {}\ar@{-}[r]|{\displaystyle Y}&\ar@{-}[u]|{} \ar@{}[ru]|{\displaystyle f} & {}\ar@{-}[l]|{}\\
    &{}\ar@{-}[u]|{\displaystyle Y}&
}
\end{gathered}\\
   &G_v =
    \begin{gathered}
   \xymatrix@=1cm{%
    &{}\ar@{-}[r]|{\displaystyle Z} & {v}\ar@{-}[d]|{\displaystyle Z}\\
    {}\ar@{-}[r]|{\displaystyle X}&\ar@{-}[u]|{\displaystyle XZ} & {}\ar@{-}[l]|{\displaystyle XZ}\\
    &{}\ar@{-}[u]|{\displaystyle X}&
}
\end{gathered}
\label{eq: VC bosonization map}
\end{eqs}

\begin{figure}[t]
    \centering
    \includegraphics[width=0.3\textwidth]{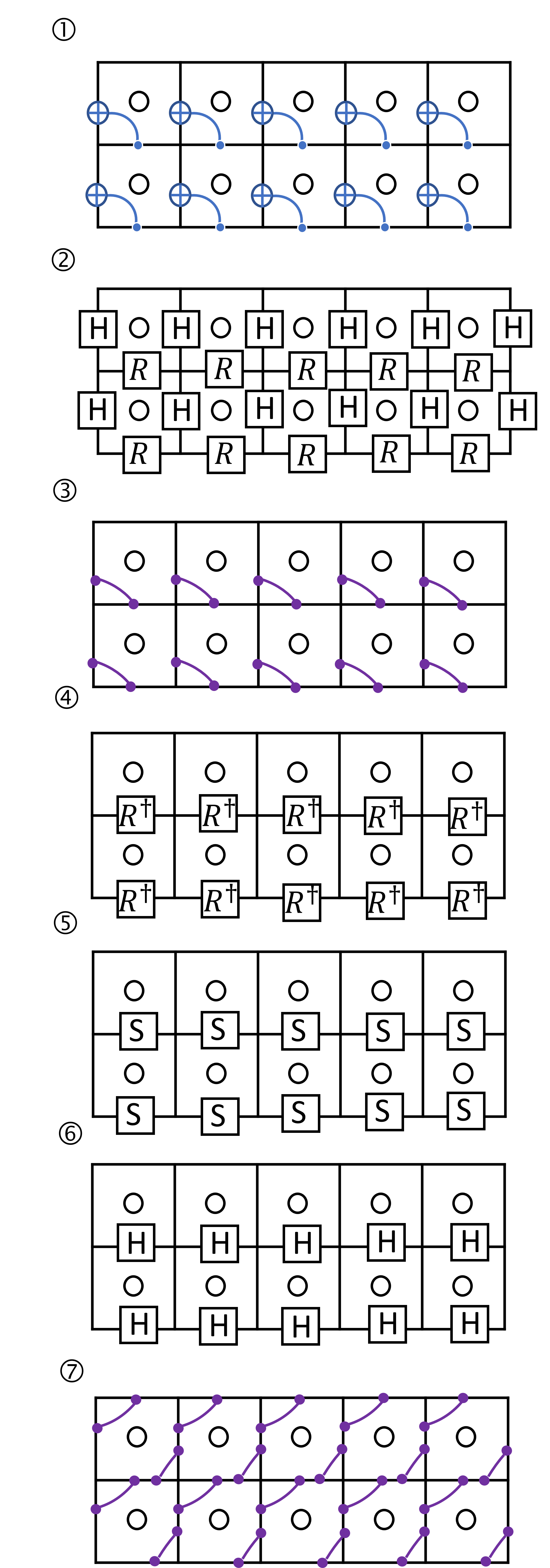}
    \caption{The finite-depth Clifford circuit to convert the Verstraete-Cirac mapping to the exact bosonization. The details of $H$, $R$, $S$ gates are discussed in Appendix~\ref{appendix:Clifford}.}
    \label{fig:VC_Clifford}
\end{figure}

The logical operators and stabilizer in Eq.~\eqref{eq: VC bosonization map} is exactly the logical operators and stabilizers of exact bosonization after a shift of Majorana fermions. If we shift the Majorana fermions in the exact bosonization as Fig.~\ref{fig:VC_shift} and re-pair them. Then we find the exact bosonization and the Verstraete-Cirac mapping are equivalent, as Fig.~\ref{fig:VC_logical}.

\begin{figure}[h]
    \centering
    \includegraphics[width=0.2\textwidth]{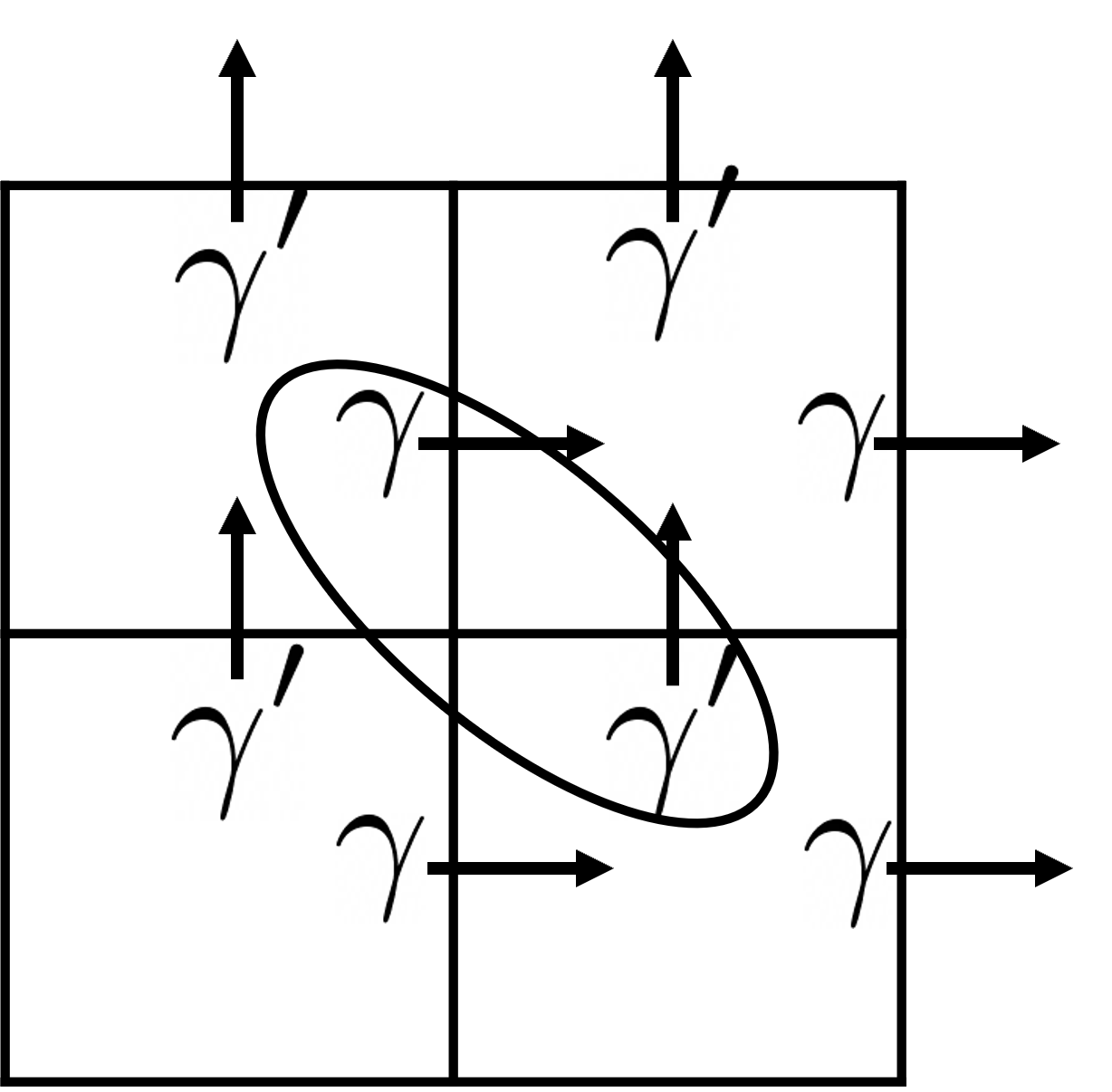}
    \caption{To match our exact bosonization to Verstraete-Cirac mapping, we shift our Majorana modes on each face as following way: 1. shift $\gamma_f'$ upward and let it be $\gamma$ on the new face; 2. shift $\gamma_{f}$ rightward and let it be $\gamma'$ on the new face.}
    \label{fig:VC_shift}
\end{figure}

\begin{figure}[h]
    \centering
    \includegraphics[width=0.3\textwidth]{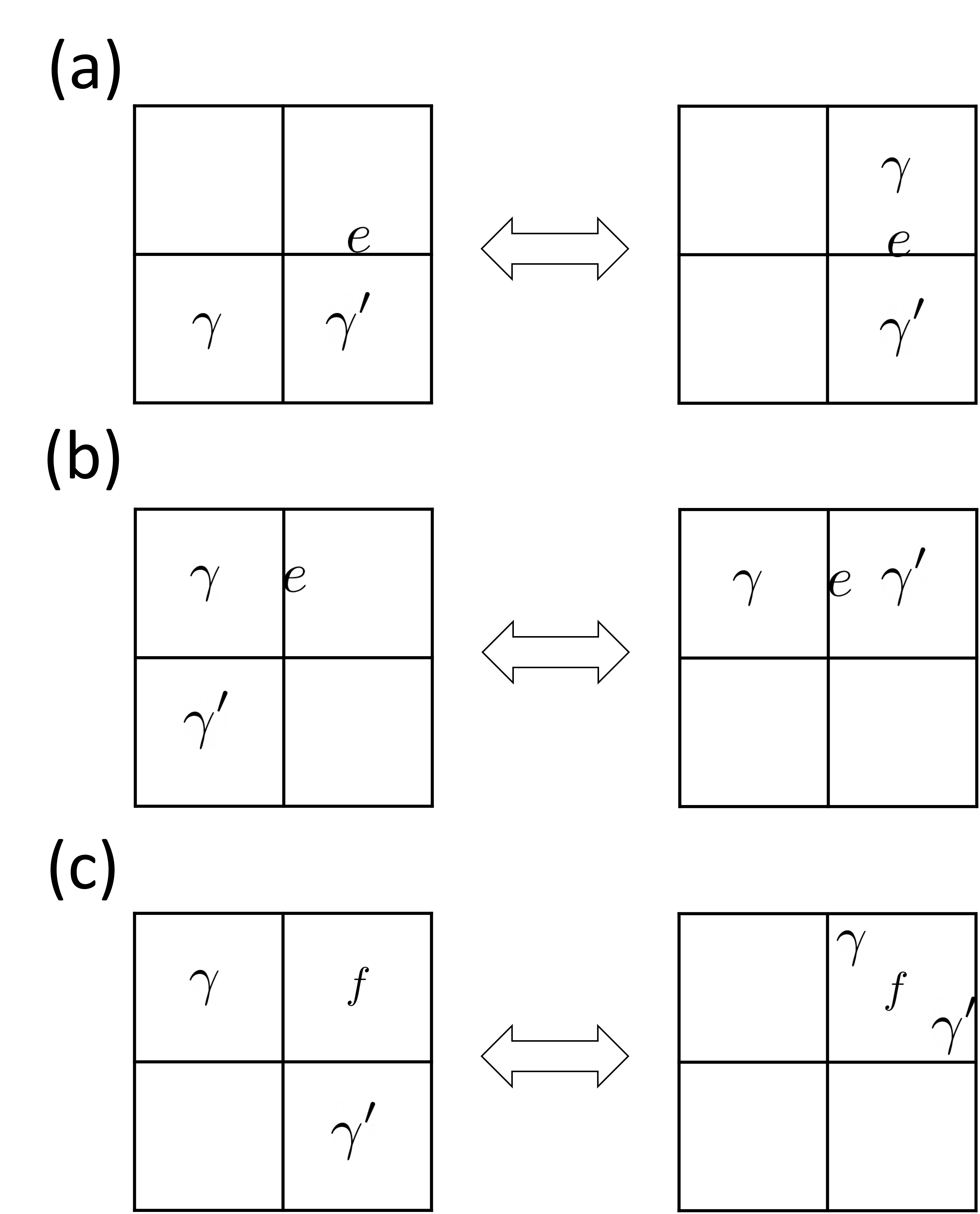}
    \caption{Correspondence of logical operators between the exact bosonization and the Verstraete-Cirac mapping}
    \label{fig:VC_logical}
\end{figure}

\subsection{Kitaev's honeycomb model}\label{sec:Kitaev_honeycomb}

The Hamiltonian of Kitaev's honeycomb model \cite{kitaev_honeycomb} can be written as
\begin{eqs}\label{eq:honeycomb_hamiltonian}
    H&=-J_x\sum_{x-\text{links}}X_j^A X_k^B-J_y\sum_{y-\text{links}}Y_j^A Y_j^B\\
    &- J_z\sum_{z-\text{links}}Z_j^A Z_k^B,
\end{eqs}
where $x$, $y$, $z$ links are shown in Fig.~\ref{fig:honeycomb_bosonization}.
The qubit at each site $j$ can be represented by four Majorana operators $b_j^x$, $b_j^y$, $b_j^z$ and $\gamma_j$ with an additional constraint $b_j^x b_j^y b_j^z\gamma_j=1$ to eliminate the redundancy at each site $j$. The Pauli matrices at each site $j$ can be represented as follows:
\begin{eqs}
    X_j=ib_j^x\gamma_j,\quad Y_j=ib_j^y\gamma_j, \quad Z_j=ib_j^z\gamma_j,
\label{eq: Kitaev qubit to majorana}
\end{eqs}
or equivalently (after multiplying by $D_j$)
\begin{eqs}
    X_j=-ib_j^yb_j^z, \quad Y=-ib_j^z b_j^x, \quad Z_j=-ib_j^xb_j^y.
\end{eqs}
Then, a free-fermion Hamiltonian
\begin{eqs}
    H=\frac{i}{2}\sum_{e_{jk}}J_{\alpha_{jk}}\gamma_j^A \gamma_k^B
\end{eqs}
is equivalent to a sector of Eq.~\eqref{eq:honeycomb_hamiltonian}, where the index $\alpha$ takes values $x$, $y$ or $z$ depending on the direction of the link $jk$. Focusing on the algebra generated by $\gamma_j$, the mapping Eq.~\eqref{eq: Kitaev qubit to majorana} can be written as \cite{CKR18}:
\begin{equation}
    i \gamma^A_j \gamma^B_k \longleftrightarrow
    \begin{cases}
        X^A_j X^B_k \quad \text{if $jk \in x$-link},\\
        Y^A_j Y^B_k \quad \text{if $jk \in y$-link},\\
        Z^A_j Z^B_k \quad \text{if $jk \in z$-link},
    \end{cases}
\label{eq: Kitaev honeycomb logical mapping}
\end{equation}
and the product of Majorana hoppings along a hexagon is proportional to identity, which gives a gauge constraint on the qubit Hilbert space. It is shown \cite{CKR18} that by embedding the honeycomb lattice into the square lattice as Fig.~\ref{fig:honeycomb}, relabeling $\gamma^A, \gamma^B$ by $\gamma_f, \gamma^\prime_f$, and performing single-qubit rotations, the complete bosonization map can be expressed as
\begin{align}
    i\times
    \begin{gathered}
    \xymatrix@=1.2cm{
    &\\
    {}\ar@{}[ur]|{\displaystyle \gamma_{L(e)}} \ar@{}[dr]|{\displaystyle \gamma'_{R(e)}} \ar@{-}[r]^{ \mathlarger e} &  \\
    &}\end{gathered}
    &\begin{gathered}
    \xymatrix{
    {}\ar@{<->}[r] &{}}
    \end{gathered}
    \hspace{0.3cm}
    \begin{gathered}
    \xymatrix@=1cm{%
    &{}\ar@{-}[d]|{\displaystyle Z} \\{}\ar@{-}[r]|{\displaystyle X_e}&
    }
    \end{gathered},
    \\[-15pt]
    i\times
    \begin{gathered}
    \xymatrix@=1.2cm{
    &{}\ar@{-}[d]^{\mathlarger e}&\\
    \ar@{}[ur]|{\displaystyle \gamma_{L(e)}}&\ar@{}[ur]|{\quad \displaystyle \gamma'_{R(e)}}&}
    \end{gathered}
    &\begin{gathered}\xymatrix{
    {}\ar@{<->}[r] &{}}
    \end{gathered}
    \begin{gathered}
    \xymatrix@=1cm{%
    {}\ar@{-}[r]|{\displaystyle Z} &{} \\{}\ar@{-}[u]|{\displaystyle X_e}& {}}
    \end{gathered},
    \label{eq: Ue on horizontal edge}
    \\[10pt]
    -i\gamma_f\gamma_f'
    &\begin{gathered}\xymatrix{
    {}\ar@{<->}[r] &{}}\end{gathered}
    \hspace{0.3cm}
    \begin{gathered}
    \xymatrix@=1cm{%
    {}\ar@{-}[r]|{\displaystyle Y}\ar@{}[dr]|{\mathlarger f} & {}\ar@{-}[d]|{\displaystyle Y} \\
    {}\ar@{-}[u]|{} & {}\ar@{-}[l]|{}}
    \end{gathered} \quad ,
\end{align}
with gauge constraints
\begin{equation}
    G_v =
    \begin{gathered}
   \xymatrix@=1cm{%
    &{}\ar@{-}[r]|{\displaystyle Z}\ar@{}[dr]|{\mathlarger f}& {}\ar@{-}[d]|{\displaystyle Z}\\
    {}\ar@{-}[r]|{\displaystyle X}&{v}\ar@{-}[u]|{\displaystyle XZ} & {}\ar@{-}[l]|{\displaystyle XZ}\\
    &{}\ar@{-}[u]|{\displaystyle X}&}
\end{gathered}
= 1.
\end{equation}
This is equivalent to the logical operators and stabilizers in Eq.~\eqref{eq: VC bosonization map} up to a shift. Therefore, it is gLU equivalent to the exact bosonization.

\begin{figure}[h]
    \centering
    \includegraphics[width=0.45\textwidth]{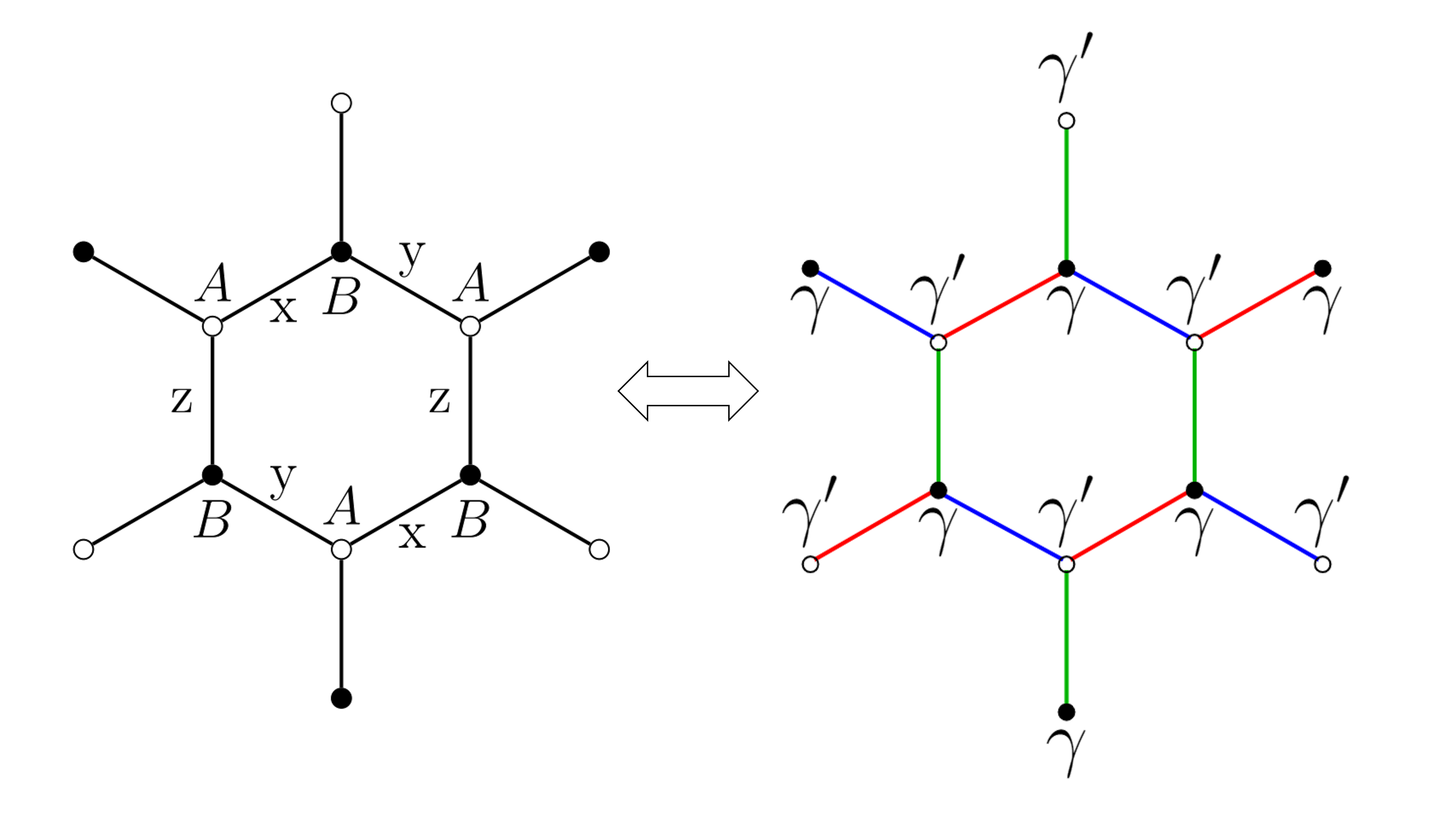}
    \caption{Kitaev's honeycomb mapping between Pauli matrices and Majorana fermions. For each link, the product of two Pauli matrices on its vertices is mapped to the product of $\gamma$ and $\gamma^\prime$ on its vertices by Eq.~\eqref{eq: Kitaev honeycomb logical mapping}.}
    \label{fig:honeycomb_bosonization}
\end{figure}

\begin{figure}[h]
    \centering
    \includegraphics[width=0.25\textwidth]{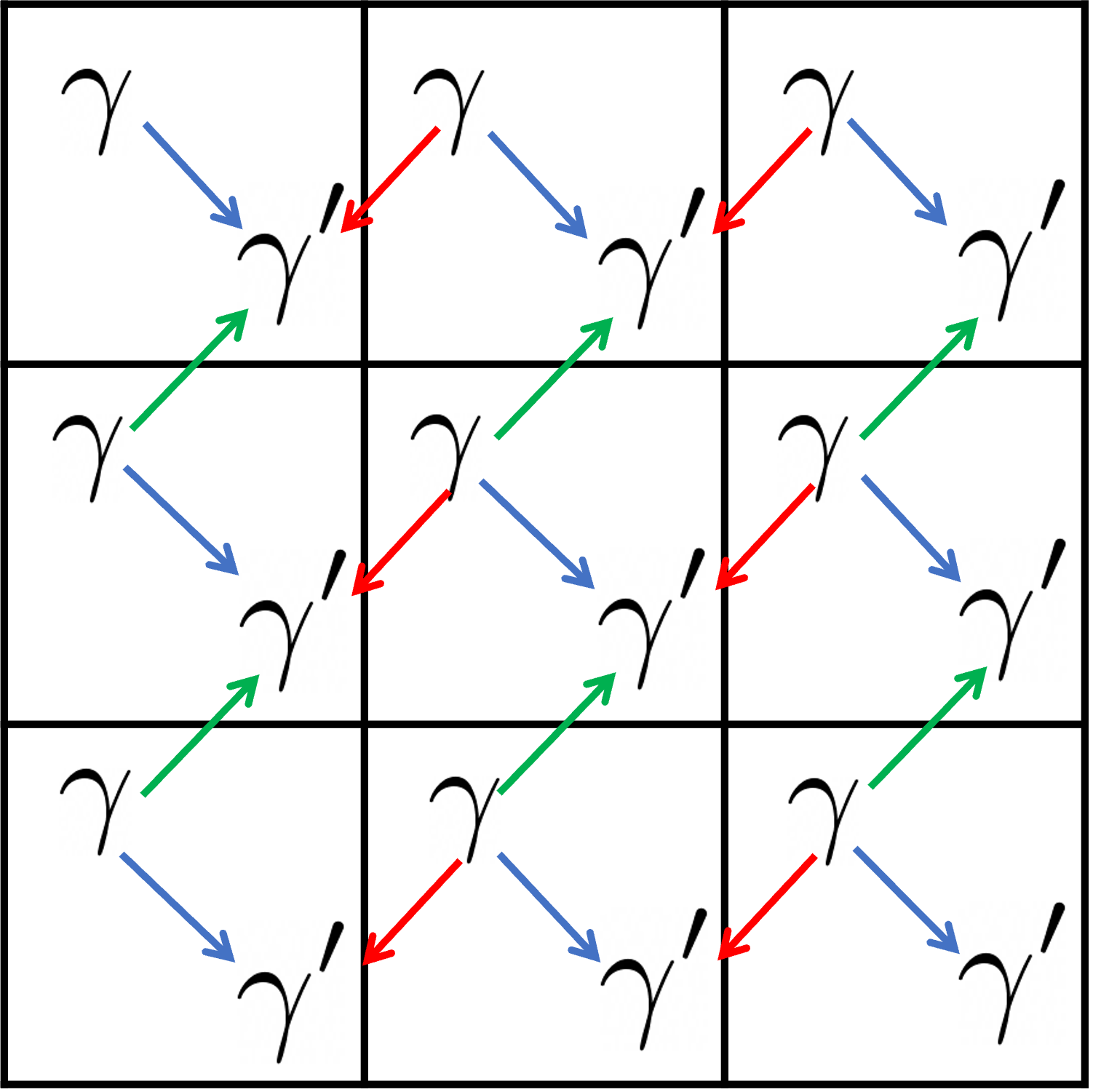}
    \caption{The embedding of the honeycomb lattice in Fig.~\ref{fig:honeycomb_bosonization} to the square lattice.}
    \label{fig:honeycomb}
\end{figure}

\subsection{Majorana loop stabilizer codes}\label{sec:Majorana_loop}
\begin{figure}[!]
    \centering
    \includegraphics[width=0.3\textwidth]{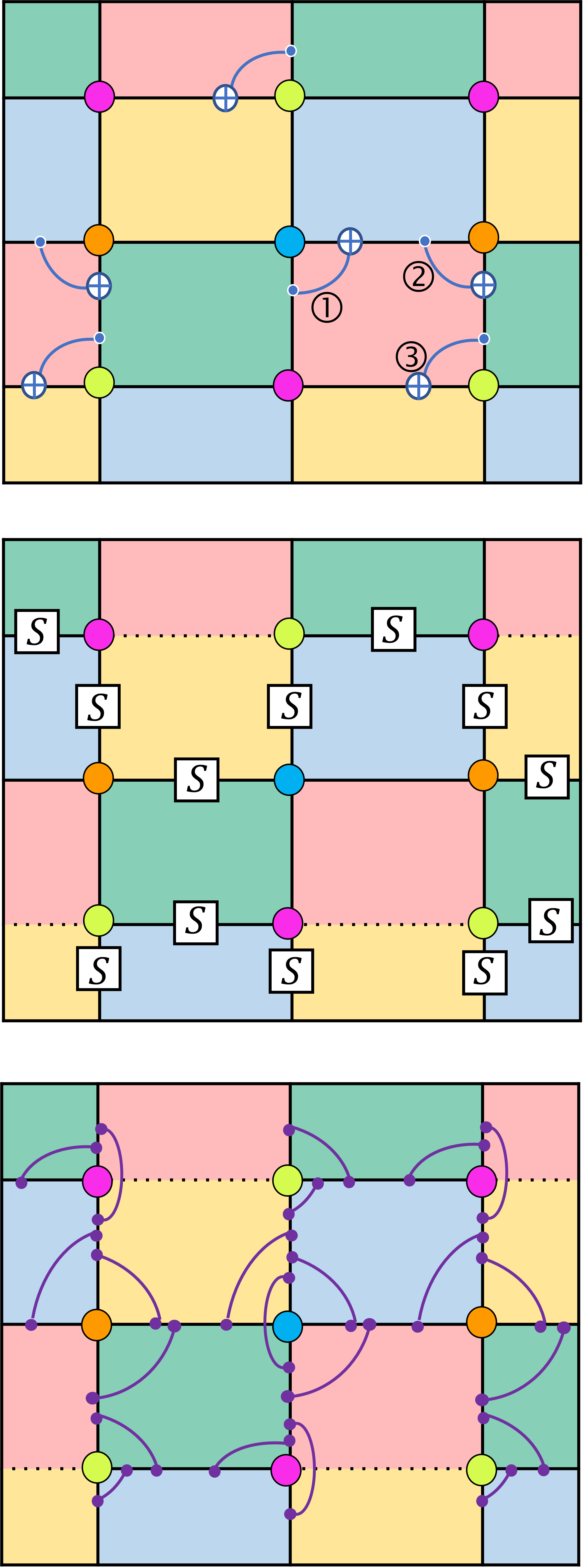}
    \caption{The finite-depth Clifford circuit for the MLSC to the exact bosonization. The first Clifford circuit will disentangle the qubits on the edges between red and yellow squares, so the edges between red and yellow squares become dashed lines in the second and third steps.}
    \label{fig:MLSC_Clifford}
\end{figure}

\begin{figure}[h]
    \centering
    \includegraphics[width=0.4\textwidth]{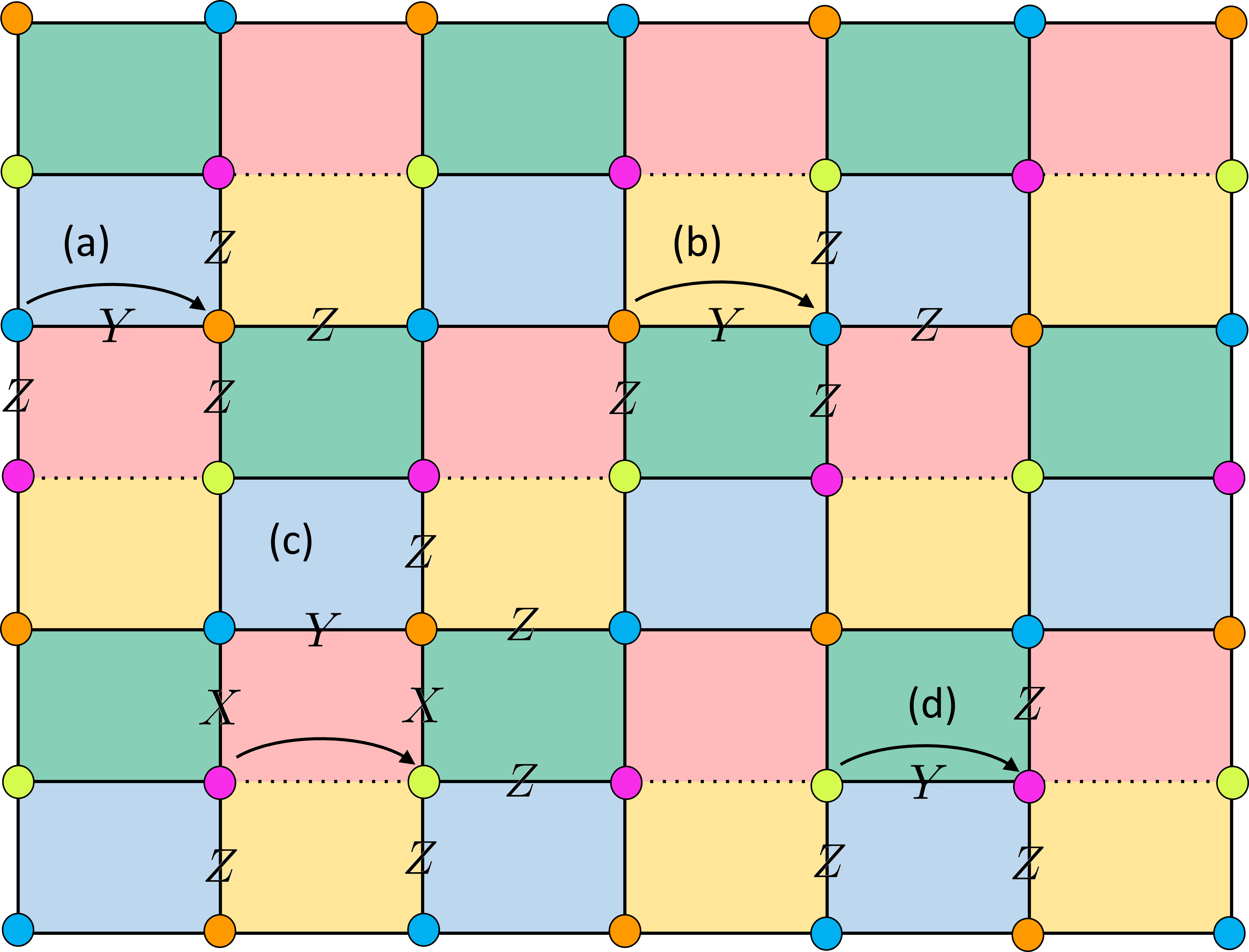}
    \caption{Horizontal hopping $i\gamma_{L(e)}\gamma_{R(e)}$ after a finite-depth gLU transformation in Fig.~\ref{fig:MLSC_Clifford}. (a),(b) are hoppings between blue and orange dots; (c),(d) are hoppings between pink and yellow dots. (a),(b), (d) are exactly the horizontal hopping in the exact bosonization, and (c) is a product of the hopping operator and the stabilizer in the exact bosonization. }
    \label{fig:MLSC_horizontal}
\end{figure}

\begin{figure}[h]
    \centering
    \includegraphics[width=0.4\textwidth]{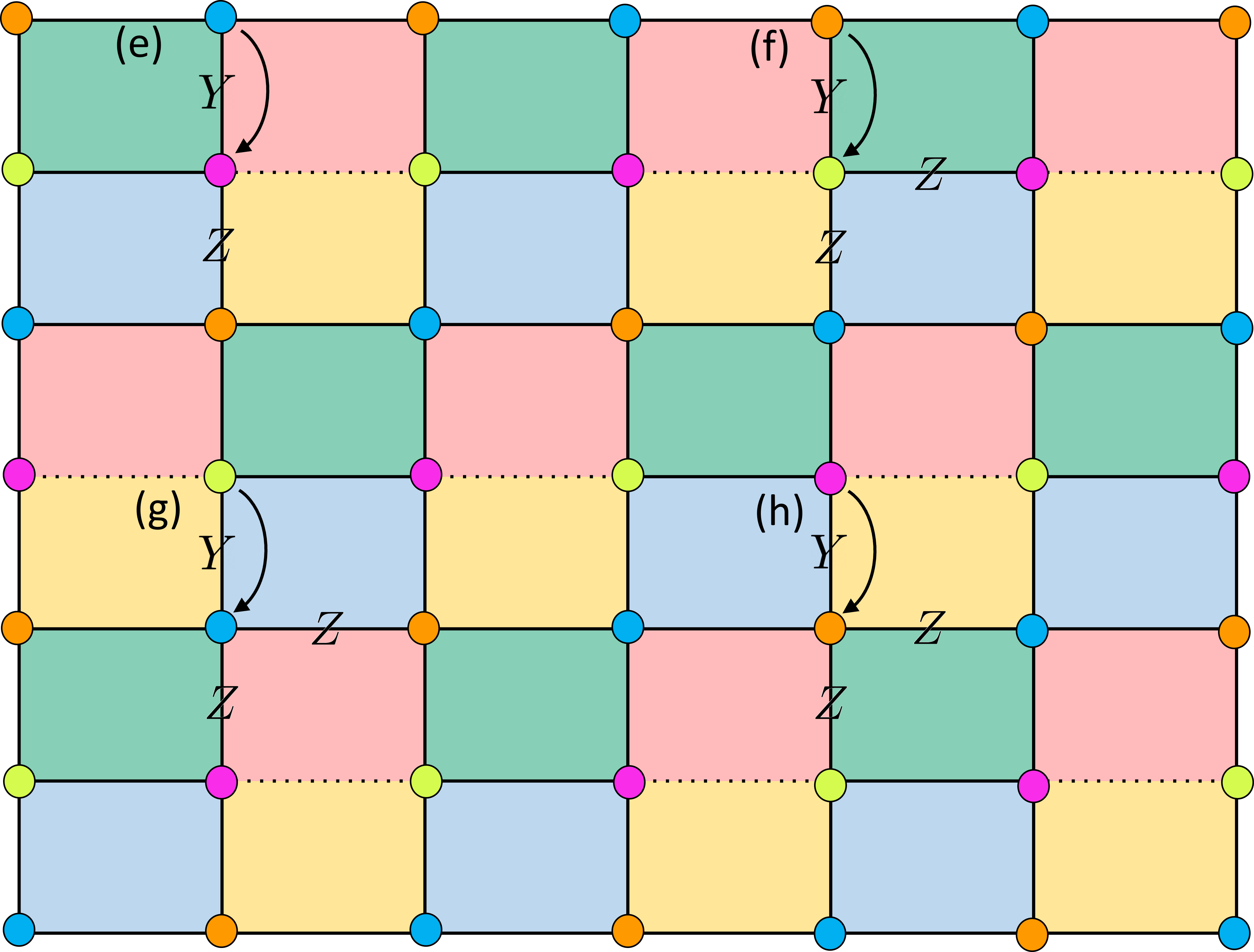}
    \caption{(e),(f),(g) and (h) are the vertical hopping $i\gamma_{L(e)}\gamma_{R(e)}$ after a finite-depth gLU transformation in Fig.~\ref{fig:MLSC_Clifford}. They match the vertical hoppings in the exact bosonization.}
    \label{fig:MLSC_vertical}
\end{figure}

In this section, we show that Majorana loop stabilizer code (MLSC) \cite{majorana_loop} is gLU equivalent to the 2d exact bosonization. Similar to BKSF, the Majorana loop stabilizer codes encode a complex fermion on vertex $v$ by qubits on edge $e$ connected to $v$. The Majorana loop stabilizer codes have fermionic hopping operation $A_e=i\gamma_{L(e)}\gamma_{R(e)}$ on each edge, fermion parity operator $P_f=-i\gamma_f\gamma_f'$ on each vertex and stabilizers $G_v$ acting on faces with different colors. We follow the same procedure described in Sec.~\ref{sec:gLU_and_bosonization}, conjugating the logical operations and stabilizers of MLSC by finite-depth Clifford circuits in Fig.~\ref{fig:MLSC_Clifford}. Then the four kinds of horizontal hoppings in MLSC reduce to the horizontal hopping in the exact bosonization (up to a stabilizer), and the same thing happens to the vertical hoppings, parity operators, and stabilizers. 

Starting from the MLSC, Fig.~\ref{fig:MLSC_horizontal} and Fig.~\ref{fig:MLSC_vertical} shows that the horizontal and vertical hoppings $i\gamma_{L(e)}\gamma_{R(e)}$ after the transformation can match the horizontal and vertical hoppings in exact bosonization. An interesting fact is that the first Clifford circuit in Fig.~\ref{fig:MLSC_Clifford} removes qubit on the edges between red and yellow squares and makes this correspondence possible.

\subsection{Connection to Jordan-Wigner transformation }\label{sec:Jordan_Wigner}

In this section, we will show that conjugating the exact bosonization by a linear-depth\footnote{The depth of the circuit scales linearly with the system size.} Clifford circuit in Fig.~\ref{fig:JW_clifford} will result in the 1d Jordan-Wigner transformation along the path in Fig.~\ref{fig:2d_Jordan_Wigner_ordering}.

For the Jordan-Wigner transformation, the qubit-fermion ratio is 1, but it is a non-local fermion-to-qubit mapping since in 1d Jordan-Wigner transformation, the vertical hopping terms are mediated by a Pauli $Z$-string between two sites. By directly applying the linear-depth gLU Clifford circuit in Fig.~\ref{fig:JW_clifford} to the logical operators of the exact bosonization, the qubits on the horizontal edges are disentangled and do not show up in the logical operators. All stabilizers become single-Pauli operators on horizontal edges and can be removed by gLU transformations. Explicitly, the logical operators after conjugation of the Clifford circuit are:

\begin{figure}[h]
\centering
\includegraphics[width=0.3\textwidth]{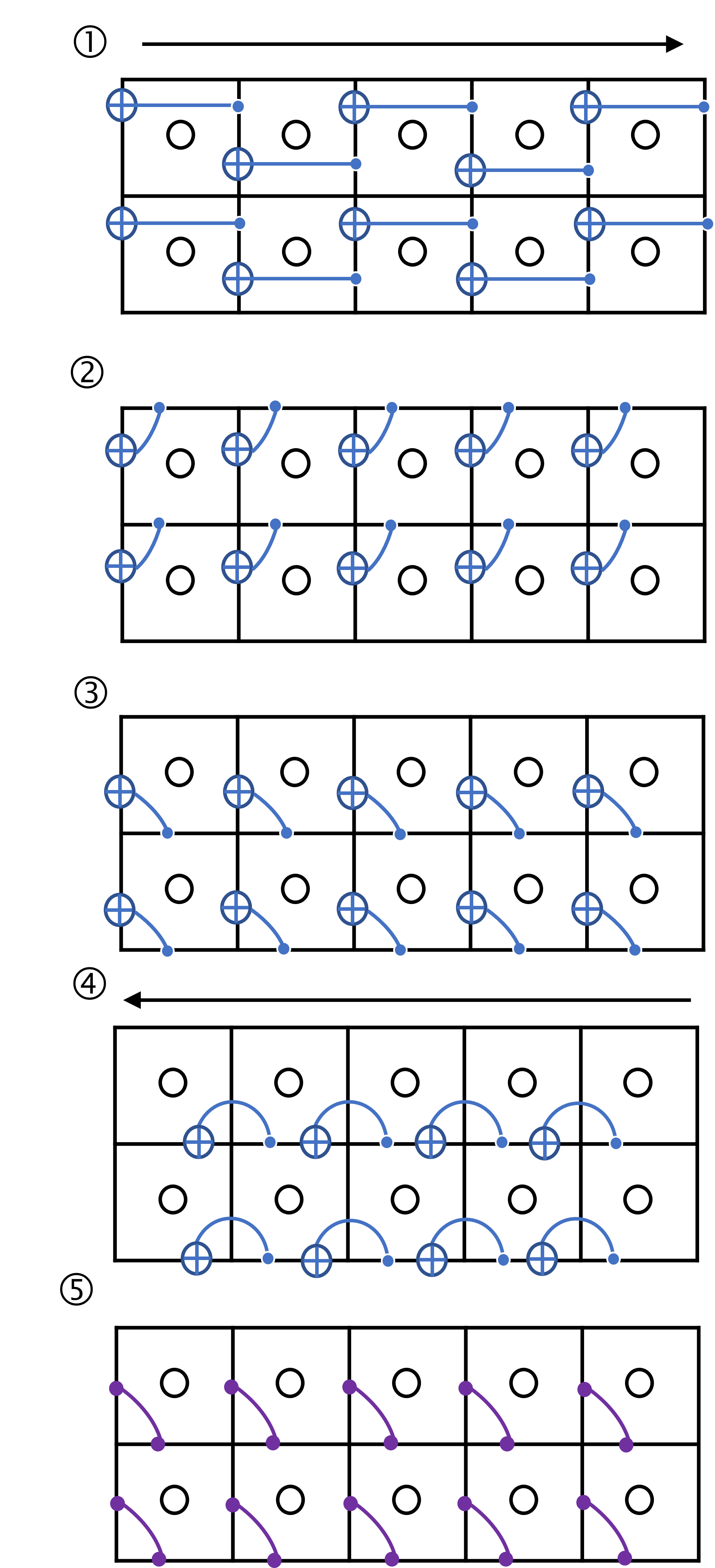}
\caption{The finite-depth Clifford circuit to convert the exact bosonization to 1d Jordan-Wigner transformation. In the first step, we ordered the system from the left to the right, then apply CNOT gate to each individual column following the above ordering. In the second step and third steps, the CNOT gates are applied simultaneously. In the fourth step, we order the system from the right to the left, then apply CNOT gate to each individual column following the right-to-left ordering. In the fifth step, CZ gates are simultaneously applied.}
\label{fig:JW_clifford}
\end{figure}

\begin{widetext}
\begin{align}\label{eq:2d_bosonization_to_JW}
    \begin{gathered}
   \xymatrix@=1cm{%
    & \\
    {}\ar@{-}[r]|{\displaystyle Z} &{}\ar@{-}[u]|{\displaystyle X_e}  }
    \end{gathered}
     &\begin{gathered}
    \xymatrix{
        {}\ar[r] &{}}
    \end{gathered} \quad
    \begin{gathered}
   \xymatrix@=1cm{%
    & &\\
    {}\ar@{-}[r]|{} \ar@{-}[u]|{\displaystyle X} \ar@{}[ru]|{\mathlarger{\mathlarger{\mathlarger{\mathlarger{\gamma}}}} }&{}\ar@{-}[u]|{\displaystyle X_e} \ar@{}[ru]|{\mathlarger{\mathlarger{\mathlarger{\mathlarger{\gamma'}}}} }&  }
    \end{gathered}
    ,
    \\
    \begin{gathered}
    \xymatrix@=1cm{%
    {}\ar@{-}[r]|{\displaystyle X_e}& \\ \ar@{-}[u]|{\displaystyle Z}& }\end{gathered}
     &\begin{gathered}
    \xymatrix{
        {}\ar[r] &{}}
    \end{gathered} \quad
    \begin{gathered}
    \vcenter{\hbox{\includegraphics[scale=.35,trim={.5cm 0cm 1.5cm 0cm},clip]{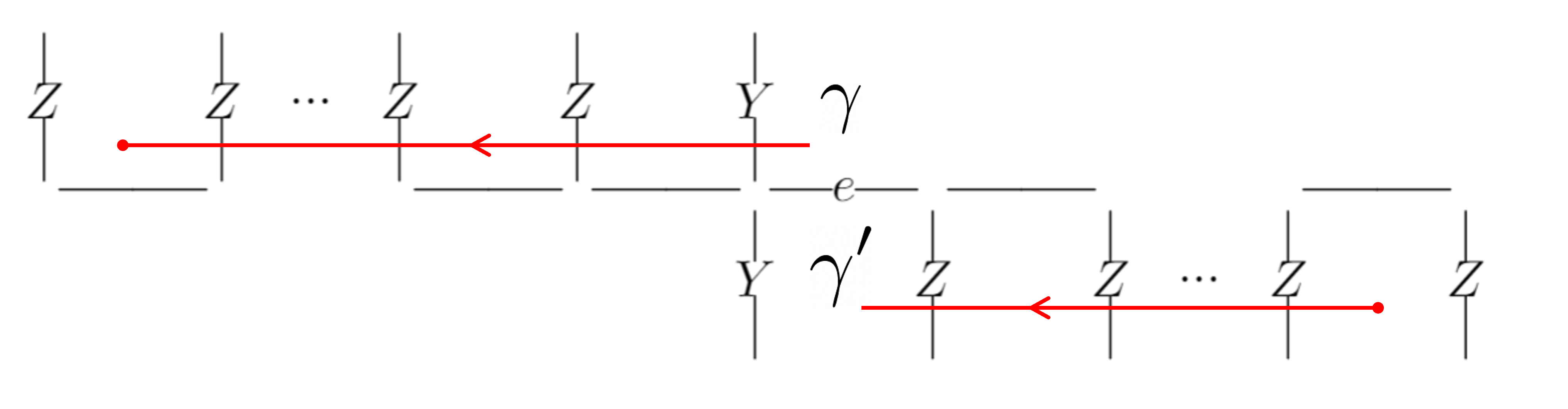}}}\end{gathered},\\
    \begin{gathered}
    \xymatrix@=1cm{%
    {}\ar@{-}[r]|{\displaystyle Z}& \\ \ar@{-}[u]|{\displaystyle Z} \ar@{-}[r]|{\displaystyle Z} \ar@{}[ur]|{\displaystyle f}& \ar@{-}[u]|{\displaystyle Z} }\end{gathered}
     &\begin{gathered}
    \xymatrix{
        {}\ar[r] &{}}
    \end{gathered} \quad
    \begin{gathered}
    \xymatrix@=1cm{%
    {}\ar@{-}[r]|{}& \\ \ar@{-}[u]|{\displaystyle Z} \ar@{-}[r]|{} \ar@{}[ur]|{\displaystyle f}& \ar@{-}[u]|{} }\end{gathered}.
\end{align}
\end{widetext}
which is exactly the 1d Jordan-Wigner transformation with the ordering chosen in Fig.~\ref{fig:2d_Jordan_Wigner_ordering}.
Hence, we can regard 1d Jordan-Wigner transformation as a special case that we remove all the qubits on the horizontal edges where the vertical hoppings are no longer local.

\begin{figure}[htb]
    \centering
    \includegraphics[width=0.25\textwidth]{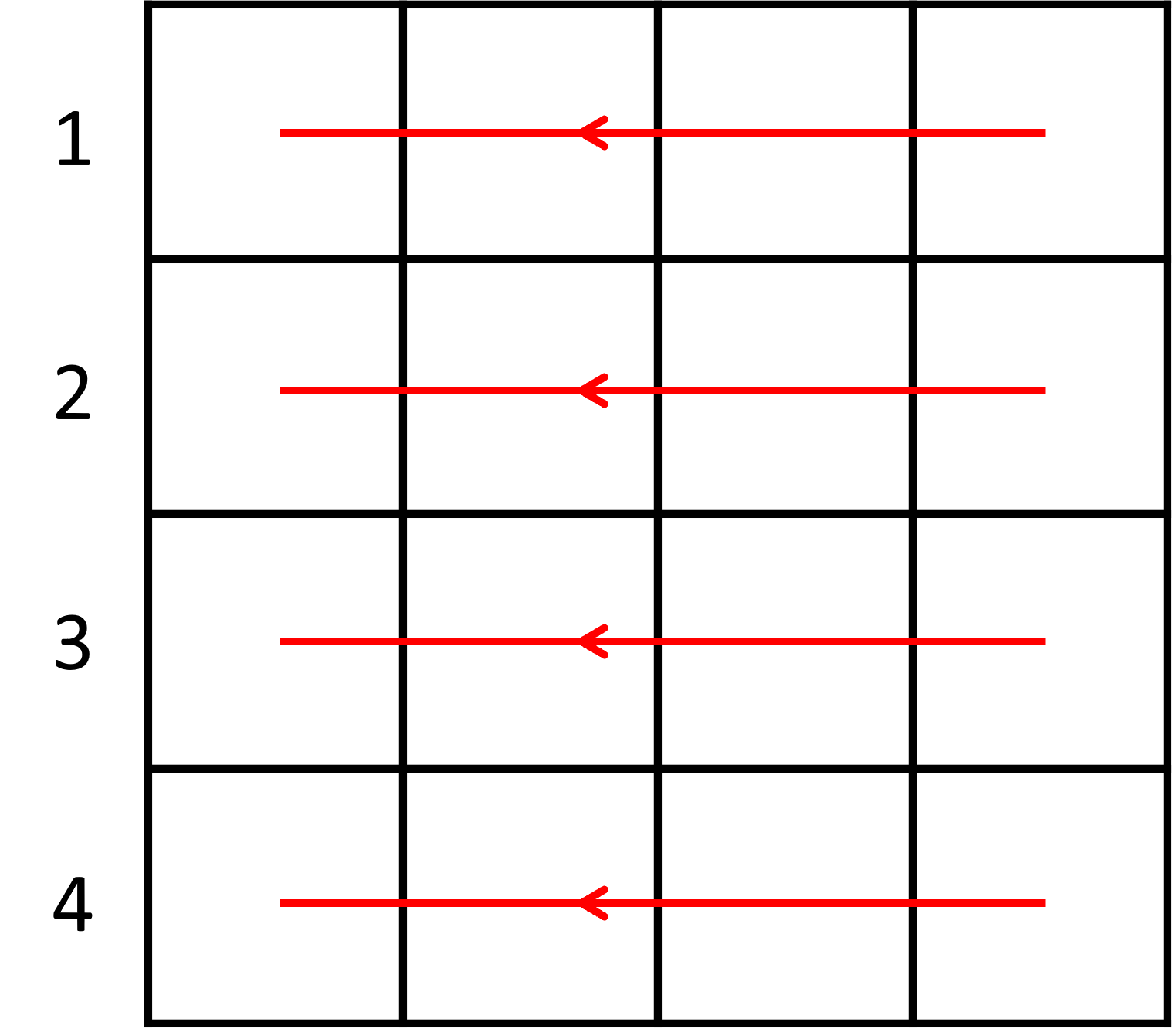}
    \caption{Space ordering of 1d Jordan-Wigner transformation}
    \label{fig:2d_Jordan_Wigner_ordering}
\end{figure}

\section*{Acknowledgement}
Y.-A.C thanks Mark Steudtner for pointing out the equivalence between the BKSF and the exact bosonization. Y.-A.C want to thank Anton Kapustin, Tyler Ellison and Nat Tantivasadakarn for useful discussions. Y.-A.C is also thankful to Bowen Yang for teaching the classification of Pauli stabilizer models in two dimensions.
Y.-A.C is supported by the JQI fellowship at the University of Maryland.
Y.X. is supported by ARO W911NF-15-1-0397 and advisor Mohammad Hafezi.

\appendix

\section{Clifford gates}\label{appendix:Clifford}
The Clifford group is defined as the group of unitaries that normalize the Pauli group. The Clifford gates are defined as elements in the Clifford group \cite{Gottesman_Clifford,gottesman1998heisenberg}. In this paper, we use single-qubit Clifford gates: $H$-gate, $S$-gate, $R$-gate.

The $H$-gate is the Hadamard gate
\begin{eqs}
    H=\frac{1}{\sqrt{2}}\begin{bmatrix}1&1\\1&-1\end{bmatrix}
\end{eqs}
that satisfies $HXH^\dagger=Z$, $HZH^\dagger=X$.

The $S$-gate is the phase gate

\begin{eqs}
    S=\begin{bmatrix}1&0\\0&i\end{bmatrix}
\end{eqs}
that satisfies $SXS^\dagger=Y$, $SYS^\dagger=-X$.

The $R$-gate is

\begin{eqs}
    R=\frac{1}{\sqrt{2}}\begin{bmatrix}1&i\\i&1\end{bmatrix}
\end{eqs}
where $RYR^\dagger=-Z$, $RZR^\dagger=Y$.

For 2-qubit Clifford gates, we choose $CNOT$, $CY$ and $CZ$ gate.
The $CNOT$ gate is
\begin{eqs}
    \text{CNOT}=\begin{bmatrix}1&0&0&0\\
    0&1&0&0\\
    0&0&0&1\\
    0&0&1&0\end{bmatrix},
\end{eqs}
where 
\begin{eqs}
&CNOT (X\otimes I)CNOT^\dagger=X\otimes X,\\ 
&CNOT(Z\otimes I)CNOT^\dagger=Z\otimes I,\\ 
&CNOT(I\otimes X)CNOT^\dagger=I\otimes X,\\ 
&CNOT(I\otimes Z)CNOT^\dagger=Z\otimes Z.
\end{eqs}

The $CY$ (controlled-$Y$) gate is
\begin{eqs}
    CY=\begin{bmatrix}1&0&0&0\\
    0&1&0&0\\
    0&0&0&-i\\
    0&0&i&0\end{bmatrix}
\end{eqs}
where
\begin{eqs}
    &CY(X\otimes I)CY^\dagger=X\otimes Y,\\
    &CY(Z\otimes I)CY^\dagger=Z\otimes I,\\
    &CY(I\otimes X)CY^\dagger=Z\otimes X,\\
    &CY(I\otimes Z)CY^\dagger=Z\otimes Z.
\end{eqs}

The $CZ$ (controlled-$Z$) gate is
\begin{eqs}
    CZ=\begin{bmatrix}1&0&0&0\\
    0&1&0&0\\
    0&0&1&0\\
    0&0&0&-1\end{bmatrix}
\end{eqs}
where
\begin{eqs}
    &CZ(X\otimes I)CZ^\dagger=X\otimes Z,\\
    &CZ(Z\otimes I)CZ^\dagger=Z\otimes I,\\
    &CZ(I\otimes X)CZ^\dagger=Z\otimes X,\\
    &CZ(I\otimes Z)CZ^\dagger=I\otimes Z.
\end{eqs}


\onecolumngrid
\clearpage
\twocolumngrid

\bibliography{bibliography.bib}

\end{document}